\documentclass[aip,jmp,author-year,a4paper,twoside,fleqn,leqno,nofootinbib,preprint]{revtex4}

\usepackage{hyperref}
\usepackage[utf8]{inputenc}
\usepackage{color}
\usepackage{calc}
\usepackage{newlfont}
\usepackage[english]{babel}
\usepackage{amssymb,amsmath,amsfonts,amsthm}
\usepackage{mathrsfs}
\usepackage{dsfont}
\usepackage[all]{xy}
\usepackage[dvips]{graphicx,psfrag}
\usepackage{enumerate}
\DeclareMathOperator{\ran}{Ran}

\renewcommand{\Im}{\mathrm{Im}}
\renewcommand{\Re}{\mathrm{Re}}
\newcommand{\de}{\mathrm{d}}

\newcommand{\braket}[2]{ \langle #1 , #2 \rangle }  
\newcommand{\norm}[1]{\bigl\lVert #1 \bigr\rVert} 
\newcommand{\ass}[1]{\left\lvert #1 \right\rvert} 
\newcommand{\dtot}[2]{\frac{\mathrm{d} #1}{\mathrm{d} #2}} 
\newcommand{\ide}[1]{\mathrm{d}#1 \,} 
\newcommand{\V}{\widetilde{V}}
\newcommand{\W}{\widetilde{W}}

\renewcommand{\v}{\widetilde{v}}
\newcommand{\U}{\widetilde{U}}
\renewcommand{\H}{\mathscr{H}}

\newcommand{\R}{\mathds{R}}

\newcommand{\nnorm}[1]{\bigl\lvert\mspace{-1.65mu}\bigl\lVert #1 \bigr\rVert\mspace{-1.65mu}\bigr\rvert}
\newcommand{\clos}[1]{#1^{\overline{\phantom{=}}}}
\theoremstyle{plain}
\newtheorem{thm}{Theorem}
\newtheorem{proposition}{Proposition}
\newtheorem{lemma}{Lemma}
\newtheorem*{lemma*}{Lemma}

\newtheorem*{corollary*}{Corollary}
\theoremstyle{definition}

\newtheorem*{definition*}{Definition}
\newtheorem*{definitions*}{Definitions}
\theoremstyle{remark}

\newtheorem*{remark*}{Remark}
\newtheorem*{remarks*}{Remarks}
\numberwithin{equation}{section}
\numberwithin{lemma}{section}
\numberwithin{proposition}{section}
\numberwithin{corollary}{section}
\numberwithin{definition}{section}

\begin{document}
\date{\today}
\title{Classical limit of the Nelson model with cut off.}
\author{Marco Falconi}
\affiliation{Dipartimento di Matematica, Università di Bologna\\Piazza di Porta San Donato 5 - 40126 Bologna, Italy}
\email[]{m.falconi@unibo.it}
\begin{abstract}
In this paper we analyze the classical limit of the Nelson model with cut off, when both non-relativistic and relativistic particles number goes to infinity. We prove convergence of quantum observables to the solutions of classical equations, and find the evolution of quantum fluctuations around the classical solution. Furthermore we analyze the convergence of transition amplitudes of normal ordered products of creation and annihilation operators between different types of initial states. In particular the limit of normal ordered products between states with a fixed number of both relativistic and non-relativistic particles yields an unexpected quantum residue: instead of the product of classical solutions we obtain an average of the product of solutions corresponding to varying initial conditions.
\end{abstract}
\maketitle
\section{Introduction.}
\label{sec:introduction}

Since the development of quantum mechanics it has been natural to analyze the connection between classical and quantum mechanical descriptions of physical systems. In particular one should expect that in some sense when quantum mechanical effects becomes negligible the system will behave like it is dictated by classical mechanics. One famous relation between classical and quantum theory dates back to early days of quantum mechanics and it is due to~\citet{ehrenfest}. This result was later developed and put on firm mathematical foundations by~\citet{MR0332046}. He proved that matrix elements of bounded functions of quantum observables between suitable coherents states (that depend on $\hslash$) converge to classical values evolving according to the expected classical equations when $\hslash\to 0$. Furthermore he also provides information about the quantum fluctuations of the system in the classical limit: their dynamics is obtained linearizing quantum evolution equation around the classical solution. His results were later generalized by~\citet{MR530915,MR605198} to bosonic systems with infinite degrees of freedom and scattering theory.~\citet*{MR2205462} applied the method described in~\citep{MR530915} to perform a partially classical limit of the Nelson model where only the number of relativistic particles goes to infinity;~\citet{MR2530155} used the same method to obtain estimates on the rate of convergence of transition amplitudes of normal ordered products of creation and annihilation operators in the mean field limit of bosonic systems. These latter results were then refined by~\citet{2011JMP} and by~\citet*{2011arXiv1103.0948C}. Mean field limits of bosonic systems has also been treated using a BBGKY hierarchy as introduced by~\citet{MR578142}~\citep[see][and references thereof contained]{MR1926667,MR2276262}, and by a method of counting introduced by~\citet{MR2821235}.

\subsection{The Hilbert space of quantum theory.}
\label{sec:hilb-space-quant}

In order to introduce the system we would like to study, we start defining the space on which the theory is set. We call it $\H$ and it is the tensor product of two symmetric Fock spaces over $L^2(\mathds{R}^3)$. Let $x_i$ for $i=1,\dotsc,p$ and $k_j$ for $j=1,\dotsc,n$ be vectors of $\mathds{R}^3$, and define
\begin{equation*}
  \mathscr{H}_{p,n}=\bigl\{ \Phi_{p,n}:\;\Phi_{p,n}(x_1,\dotsc,x_p;k_1,\dotsc,k_n)\in L^2(\mathds{R}^{3p+3n})  \bigr\}\; ,
\end{equation*}
where $\Phi_{p,n}$ is separately symmetric with respect to the first $p$ and the last $n$ variables. Then
\begin{equation*}
  \mathscr{H}=\bigoplus_{p,n=0}^\infty\mathscr{H}_{p,n}\; .
\end{equation*}
The vacuum state will be denoted by $\Omega$. We will use freely the following properties of the tensor product of Hilbert spaces:
\begin{gather*}
  \mathscr{H}_{p,n}=\mathscr{H}_{p,0}\otimes\mathscr{H}_{0,n}\\
\mathscr{F}_s(p)\equiv\bigoplus_{p=0}^\infty\mathscr{H}_{p,0}\; ,\; \mathscr{F}_s(n)\equiv\bigoplus_{n=0}^\infty\mathscr{H}_{0,n}\\
\mathscr{H}_p\equiv\mathscr{H}_{p,0}\otimes \bigoplus_{n=0}^\infty\mathscr{H}_{0,n}\\
\mathscr{H}=\mathscr{F}_s(p)\otimes\mathscr{F}_s(n)=\bigoplus_{p=0}^\infty\mathscr{H}_p\; .
\end{gather*}

We will call $\psi^{\#}(x)$ the annihilation and creation operator-valued distributions corresponding to $\mathscr{F}_s(p)$, $a^{\#}(k)$ the ones corresponding to $\mathscr{F}_s(n)$:
\begin{equation*}
    (\psi(x)\Phi)_{p,n}(x_1,\dotsc,x_p;k_1,\dotsc,k_n)=\sqrt{p+1}\, \Phi_{p+1,n}(x,x_1,\dotsc,x_p;k_1,\dotsc,k_n)\mspace{117mu}
\end{equation*}
\begin{equation*}
   (\psi^*(x)\Phi)_{p,n}(x_1,\dotsc,x_p;k_1,\dotsc,k_n)=\frac{1}{\sqrt{p}}\sum_{i=1}^p\delta(x-x_i)\Phi_{p-1,n}(x_1,\dotsc,\hat{x}_i,\dotsc,x_p;k_1,\dotsc,k_n)\; ;
\end{equation*}
\begin{equation*}
    (a(k)\Phi)_{p,n}(x_1,\dotsc,x_p;k_1,\dotsc,k_n)=\sqrt{n+1}\, \Phi_{p,n+1}(x_1,\dotsc,x_p;k,k_1,\dotsc,k_n) \mspace{122mu}
\end{equation*}
\begin{equation*}
  (a^*(k)\Phi)_{p,n}(x_1,\dotsc,x_p;k_1,\dotsc,k_n)=\frac{1}{\sqrt{n}}\sum_{j=1}^n\delta(k-k_j)\Phi_{p,n-1}(x_1,\dotsc,x_p;k_1,\dotsc,\hat{k}_j,\dotsc,k_n)\; ,
\end{equation*}
where $\hat{x}_i$ or $\hat{k}_j$ means such variable is missing.
They satisfy the canonical commutation relations
\begin{gather*}
  [\psi(x),\psi^*(x')]=\delta(x-x')\; ,\; [\psi(x),\psi(x')]=[\psi^*(x),\psi^*(x')]=0\\
    [a(k),a^*(k')]=\delta(k-k')\; ,\; [a(k),a(k')]=[a^*(k),a^*(k')]=0\; ;
\end{gather*}
obviously also $[\psi^\#(x),a^\#(k)]=0$ since they correspond to distinct Fock spaces. We will also use the following abbreviations:
\begin{equation*}
  \int\de X_p\equiv \int_{\mathds{R}^3}\ide{x_1}\dotsm\int_{\mathds{R}^3}\de x_p\; ,\; \int\de K_n\equiv \int_{\mathds{R}^3}\ide{k_1}\dotsm\int_{\mathds{R}^3}\de k_n\; ;
\end{equation*}
\begin{equation*}
 \psi^\#(X_p)\equiv\prod_{i=1}^p\psi^\#(x_i)\; ,\; a^\#(K_n)\equiv\prod_{j=1}^n a^\#(k_j)\; .
\end{equation*}
For any $f,g\in L^2(\mathds{R}^3)$ we define the annihilation and creation operators:
\begin{gather*}
  \psi^\#(f)=\int\ide{x}f(x)\psi^\#(x)\; ;\; a^\#(g)=\int\ide{k}g(k)a^\#(k)\; .
\end{gather*}
We also have a particle number operator corresponding to each Fock subspace, we will call them $N_1$ and $N_2$ and are defined as follows
\begin{gather*}
(N_1\Phi)_{p,n}=p\: \Phi_{p,n}\; ,\; (N_2\Phi)_{p,n}=n\: \Phi_{p,n}\; ;\\
N= N_1+N_2\; .  
\end{gather*}
The corresponding domains of definition are respectively $D(N_1)$, $D(N_2)$ and $D(N)$.

\subsection{The Nelson Hamiltonian.}
\label{sec:nels-hamilt-fock}

We are now able to introduce the Hamiltonian that describes the dynamics of our system. It has been widely studied in mathematical physics~\citep[see for example][]{EugeneP1962219,nelson:1190,MR1809881,MR1943190}. It was introduced to describe a theory of non-relativistic nucleons (bosonic or fermionic) interacting with a meson field; recent developments in the study of quantum optics showed this model is also useful to describe systems of bosons interacting with radiation.

In the language of second quantization, the non-relativistic boson particles are described by a Schrödinger field on $\mathscr{F}_s(p)$, with mass $M>0$; the relativistic boson field is described by a Klein-Gordon field on $\mathscr{F}_s(n)$ with mass $\mu\geq 0$. So the free part of the Hamiltonian $H_0$ can be written as
\begin{equation*}
      H_0=\frac{1}{2M}\int\ide{x}(\nabla\psi)^*(x)\nabla\psi(x)+\int\ide{k}\omega(k)a^*(k)a(k)
\end{equation*}
with $\omega(k)=\sqrt{k^2+\mu^2}$. The interaction occurs between the two different species, with a cutoff for large momenta of the relativistic field, and coupling constant $\lambda>0$:
\begin{gather*}
  H_I=\lambda\int\ide{x}\varphi(x)\psi^*(x)\psi(x)\; ,\\
   \varphi(x)=\int\frac{\de k}{(2\pi)^{3/2}}\frac{1}{(2\omega)^{1/2}}\chi(k)\Bigl(a(k)e^{ikx}+a^*(k)e^{-ikx}\Bigr)\text{ with}\\
\chi(k)=\left\{
  \begin{aligned}
    1&\text{ if $\ass{k}\leq\sigma$}\\
    0&\text{ if $\ass{k}>\sigma$}
  \end{aligned}
\right.\; .
\end{gather*}
Finally the Nelson Hamiltonian is the sum of $H_0$ and $H_I$:
\begin{equation*}
      H=\frac{1}{2M}\int\ide{x}(\nabla\psi)^*(x)\nabla\psi(x)+\int\ide{k}\omega(k)a^*(k)a(k)+\lambda\int\ide{x}\varphi(x)\psi^*(x)\psi(x)\; .
\end{equation*}
For further details on the physical significance of such system see \citet{EugeneP1962219}.

We call $U(t)$ the unitary evolution generated by $H$, $U_0(t)$ the one generated by $H_0$:
\begin{gather*}
  U(t)=\exp\{-itH\}\; ;\; U_0(t)=\exp\{-itH_0\}\; \forall t\in\mathds{R}\; .
\end{gather*}
This dynamics leaves invariant the number of non-relativistic particles $N_1$, and so also each subspace $\mathscr{H}_p$. However since we want to consider the mean field limit of the system (where the number of non-relativistic particles goes to infinity), we need to consider the whole space $\mathscr{H}$.
\subsection{The classical equations.}
\label{sec:classical-equations}

Classical dynamics is described by a semi-linear Schrödinger/Klein-Gordon system of equations:
\begin{equation}\label{eq:1}
  \left\{
  \begin{aligned}
    \Bigl(i\partial_t+\frac{1}{2M} \Delta\Bigr) u&= (2\pi)^{-3/2}(\mathcal{F}^{-1}(\chi) * A) u\\
    (\partial_t^2-\Delta +\mu^2)A&= -(2\pi)^{-3/2}\mathcal{F}^{-1}(\chi)* \ass{u}^2\; ;
  \end{aligned}\right .
\end{equation}
where $u$ is a complex-valued and $A$ a real-valued function of $\mathds{R}^3$, and we use the following convention for the Fourier transform in $L^2(\mathds{R}^3)$:
\begin{align*}
   \mathcal{F}(g)(k)&=\frac{1}{(2\pi)^{3/2}}\int\ide{x}e^{-ikx}g(x) \; .
 \end{align*}
It is useful to rewrite equation~\eqref{eq:1} as
\begin{equation}\label{eq:2}
  \left\{
  \begin{aligned}
    i\partial_t u&=-\frac{1}{2M} \Delta u + (2\pi)^{-3/2}(\mathcal{F}^{-1}(\chi) * A) u\\
    i\partial_t\alpha&= \omega \alpha +(2\omega)^{-1/2} \chi \mathcal{F}(\ass{u}^2)\; ;
  \end{aligned}\right .
\end{equation}
and $A=\mathcal{F}((2\omega)^{-1/2}\bar{\alpha})+\mathcal{F}^{-1}((2\omega)^{-1/2}\alpha)$. Existence and uniqueness of a solution, in a suitable sense, to the $L^2(\mathds{R}^3)$-Cauchy problem associated with~\eqref{eq:2} is discussed in section~\ref{sec:classical-theory}.

\subsection{The classical limit.}
\label{sec:classical-limit}

We want to study the behavior of the quantum system when the number of both relativistic and non-relativistic particles is very large. It is expected that in such limit the dynamics of a non-relativistic particle would be coupled to the one of the classical Klein-Gordon field, as dictated by classical equations~\eqref{eq:1}.

One has to choose initial states suitable to perform the limit. Let $i$, $j$ $\in\mathds{N}$ be respectively the number of non-relativistic and relativistic particles (so that the classical limit corresponds to $i$, $j\to\infty$), and $u_0$, $\alpha_0\in L^2(\mathds{R}^3)$ such that $\norm{u_0}_2=\norm{\alpha_0}_2=1$; we define the following vectors of $\mathscr{H}$:
\begin{equation}
  \label{eq:4}
  \begin{aligned}
  \Lambda=C(\sqrt{i}\,u_0,\sqrt{j}\,\alpha_0)\Omega\; ;\; \Psi=u_0^{\otimes_{i}}\otimes C_n(\sqrt{j}\,\alpha_0)\Omega\in \mathscr{H}_{i}\; ;\; \Theta=u_0^{\otimes_{i}}\otimes\alpha_0^{\otimes_{j}}\in \mathscr{H}_{i,j}
  \end{aligned}
\end{equation}
where $C$ and $C_n$ are the Weyl operators defined for all $u$, $\alpha\in L^2(\mathds{R}^3)$ as following
\begin{gather*}
  C_p(u)=\exp\bigl\{\bigl(\psi^*(u)-\psi(\bar{u})\bigr)^{\overline{\phantom{=}}}\bigr\}\text{ on }\mathscr{F}_s(p)\\
  C_n(\alpha)=\exp\bigl\{\bigl(a^*(\alpha)-a(\bar{\alpha})\bigr)^{\overline{\phantom{=}}}\bigr\}\text{ on }\mathscr{F}_s(n)\\
C(u,\alpha)=C_p(u)\otimes C_n(\alpha)=\exp\bigl\{\bigl(\psi^*(u)-\psi(\bar{u})+a^*(\alpha)-a(\bar{\alpha})\bigr)^{\overline{\phantom{=}}}\bigr\}\text{ on }\mathscr{H}\text{ ,}
\end{gather*}
$\clos{A}$ stands for the closure of any closable operator $A$. The properties of Weyl operators we will use the most are stated in section~\ref{sec:weyl-operators}. Observe that at fixed time $\braket{\Xi}{\psi^\#(f)\Xi}\sim \sqrt{i}$, $\braket{\Xi}{a^\#(f)\Xi}\sim \sqrt{j}$ for all $\Xi\in \{\Lambda,\Psi,\Theta\}$, in accordance with the interpretation of $i$ and $j$ as number of non-relativistic and relativistic particles respectively.

Consider now time evolution in the Heisenberg picture. The quantum variables $\psi^\#(t,x)$ and $a^\#(t,k)$ obey the following evolution equations:
\begin{equation}\label{eq:5}
\left\{
  \begin{aligned}
    i\partial_t\psi&=-\frac{1}{2M} \Delta\psi +\lambda \varphi \psi\\
    i\partial_ta&= \omega a +\lambda \frac{\chi}{\sqrt{2\omega}} \mathcal{F}(\psi^*\psi)
  \end{aligned}\right .\; .
\end{equation}
In order to obtain a non trivial limiting equation for~\eqref{eq:5} when $i$, $j\to\infty$ we need to relate $\lambda$ to $i$ and $j$, according to
\begin{equation}
  \label{eq:3}
 i=j=\lambda^{-2}\; .
\end{equation}
So the mean field limit is also a weak coupling limit $\lambda\to 0$.

Using the quantum evolution operator $U(t)$ we can write the solution of~\eqref{eq:5}, with initial condition
\begin{equation*}
\left\{
  \begin{aligned}
    \psi(0,x)&=\psi(x)\\
    a(0,k)&= a(k)
  \end{aligned}\right .
\end{equation*}
as
\begin{equation}
  \label{eq:6}
\left\{
  \begin{aligned}
    \psi(t,x)&=U^*(t)\psi(x)U(t)\\
    a(t,k)&=U^*(t)a(k)U(t)
  \end{aligned}\right .\; .
\end{equation}
As discussed above $\psi^\#(0)\sim a^\#(0)\sim\lambda^{-1}$ when averaged over the vectors defined in~\eqref{eq:4}, so $\lambda\psi^\#(t,x)$ and $\lambda a^\#(t,k)$ are expected to have finite limit when $\lambda\to 0$. In fact we prove that their average converge to the solution of classical equations. We also extend the convergence results to normal ordered products of creation and annihilation operators (each one again multiplied by $\lambda$ to ensure convergence). These results are discussed in section~\ref{sec:stat-main-results}; in order to do that we define, for all $\Phi\in\mathscr{H}$, $\Xi\in\{\Lambda,\Psi,\Theta\}$ (see equation~\eqref{eq:4}):
\begin{equation}
  \label{eq:8}
  \begin{aligned}
        \langle\lambda \psi^\#(t,\cdot)\rangle_{C\Phi}&\equiv\braket{C(u_0/\lambda,\alpha_0/\lambda)\Phi}{U^*(t)\lambda\psi^\#(\cdot)U(t) C(u_0/\lambda,\alpha_0/\lambda)\Phi}\\
    \langle\lambda a^\#(t,\cdot)\rangle_{C\Phi}&\equiv\braket{C(u_0/\lambda,\alpha_0/\lambda)\Phi}{U^*(t)\lambda a^\#(\cdot)U(t) C(u_0/\lambda,\alpha_0/\lambda)\Phi}\; ,
  \end{aligned}
\end{equation}
\begin{equation}
  \label{eq:9}
  \begin{split}
      \langle \psi^*(q)\psi(r)a^*(h)a(l) \rangle_\Xi(t)\equiv \lambda^\delta\braket{\Xi}{U^*(t)\psi^*(X_q)\psi(Y_r)a^*(K_h)a(M_l)U(t)\Xi}\\
=\lambda^\delta\braket{\Xi}{U^*(t)\prod_{i=1}^q\psi^*(x_i) \prod_{i'=1}^r\psi(y_{i'}) \prod_{j=1}^h a^*(k_j) \prod_{j'=1}^l a(m_{j'})U(t)\Xi}\; ;
  \end{split}
\end{equation}
where $\delta=q+r+h+l$.

We have not yet considered fluctuations around the classical solution. If we write $H$ as a function of $\lambda\psi$ and $\lambda a$ we have that
\begin{equation*}
  H=\lambda^{-2}h(\lambda\psi,\lambda a)\; ,
\end{equation*}
with
\begin{equation*}
  h(\psi,a)=\frac{1}{2M}\int\ide{x}(\nabla\psi)^*\nabla\psi+\int\ide{k}\omega a^*a+\int\ide{x}\varphi\psi^*\psi\; .
\end{equation*}
Let $(u,\alpha)$ be the classical solution, and expand $h$ around $(u,\alpha)$:
\begin{equation*}
  \begin{split}
      h(\lambda\psi,\lambda a)=h(u,\alpha)+h_1(\lambda\psi-u,\lambda a -\alpha)+h_2(\lambda\psi-u,\lambda a -\alpha)+h_3(\lambda\psi-u,\lambda a -\alpha)
  \end{split}
\end{equation*}
with
\begin{equation*}
    h_1(\psi,a)=-\frac{1}{2M}\int\ide{x}\Delta u \psi^*+\int\ide{k}\omega\alpha a^*+\int\ide{x}\bigl(\frac{1}{2}\ass{u}^2\varphi+((2\pi)^{-3/2}\mathcal{F}^{-1}(\chi) * A)u\psi^*\bigr)+\text{h.c.}\mspace{33mu}
\end{equation*}
\begin{equation*}
  h_2(\psi,a)=\frac{1}{2M}\int\ide{x}(\nabla\psi)^*\nabla\psi+\int\ide{k}\omega a^*a+\biggl[\int\ide{x}\bigl(\frac{1}{2}((2\pi)^{-3/2}\mathcal{F}^{-1}(\chi) * A)\psi^*\psi+u\varphi\psi^*\bigr)+\text{h.c.}\biggr]
\end{equation*}
\begin{equation*}
  h_3(\psi,a)=\int\ide{x}\varphi\psi^*\psi\; .\mspace{525mu}
\end{equation*}
Now define
\begin{equation*}
  \begin{aligned}
  h_{k,\psi}(\psi,a)&=[\psi,h_k(\psi,a)]\; ;\; h_{k,a}(\psi,a)&=[a,h_k(\psi,a)]\text{ with $k=1$, $2$, $3$.}
  \end{aligned}
\end{equation*}
Since $(u,\alpha)$ satisfy~\eqref{eq:2}, equations~\eqref{eq:5} then could be rewritten as
\begin{equation}\label{eq:10}
  \left\{
    \begin{aligned}
      i\partial_t(\psi-u_\lambda)=&h_{2,\psi}(\psi - u_\lambda,a -\alpha_\lambda) +\lambda h_{3,\psi}(\psi - u_\lambda,a -\alpha_\lambda)\\
      i\partial_t(a-\alpha_\lambda)=&h_{2,a}(\psi - u_\lambda,a -\alpha_\lambda)+\lambda h_{3,a}(\psi - u_\lambda,a -\alpha_\lambda)
    \end{aligned}
\right.
\end{equation}
where
\begin{equation*}
  \left\{
    \begin{aligned}
      u_\lambda&=\frac{1}{\lambda}u\\
      \alpha_\lambda&=\frac{1}{\lambda}\alpha
    \end{aligned}
\right.\; .
\end{equation*}
When $\lambda\to 0$,~\eqref{eq:10} describes the evolution of quantum fluctuations. In order to take the limit it is necessary to define new variables with initial conditions independent of $\lambda$:
\begin{equation*}
\left\{
    \begin{aligned}
      \theta(t)&\equiv C(u_\lambda(s),\alpha_\lambda(s))^*(\psi(t)-u_\lambda(t))C(u_\lambda(s),\alpha_\lambda(s))\\
      b(t)&\equiv C(u_\lambda(s),\alpha_\lambda(s))^*(a(t)-\alpha_\lambda(t))C(u_\lambda(s),\alpha_\lambda(s))
    \end{aligned}
\right.
\end{equation*}
such that $\theta(s,x)=\psi(x)$ and $b(s,k)=a(k)$. Then $\theta(t)$ and $b(t)$ satisfy the Cauchy problem
\begin{equation}\label{eq:11}
  \left\{
    \begin{aligned}
      i\partial_t\theta=&h_{2,\psi}(\theta,b) +\lambda h_{3,\psi}(\theta,b)\\
      i\partial_tb=&h_{2,a}(\theta,b)+\lambda h_{3,a}(\theta,b)
    \end{aligned}
\right.
\end{equation}
\begin{equation*}
  \left\{
    \begin{aligned}
      \theta(s)&=\psi\\
      b(s)&=a
    \end{aligned}
\right.\; .
\end{equation*}
The solution of~\eqref{eq:11} is given by
\begin{equation*}
  \left\{
    \begin{aligned}
      \theta(t,x)=&\, W^*(t,s)\psi(x) W(t,s)\\
      b(t,k)=&\, W^*(t,s)a(k) W(t,s)
    \end{aligned}
\right.
\end{equation*}
with
\begin{equation}
  \label{eq:12}
  \begin{aligned}
        W(t,s)&=C^*(u_\lambda(t),\alpha_\lambda(t))U(t-s)C(u_\lambda(s),\alpha_\lambda(s))e^{i\Lambda(t,s)}\\
      \Lambda(t,s)&=-\frac{1}{2}(2\pi)^{-3/2}\lambda^{-2}\int_s^t\de{t'}\int\ide{x}(\mathcal{F}^{-1}(\chi) * A)(t')\bar{u}(t')u(t')\; .
  \end{aligned}
\end{equation}
Taking the limit $\lambda\to 0$ in~\eqref{eq:11} we obtain the equations defining the fluctuations in the classical limit:
\begin{equation}\label{eq:13}
  \left\{
    \begin{aligned}
      i\partial_t\psi_2=&h_{2,\psi}(\psi_2,a_2)\\
      i\partial_t a_2=&h_{2,a}(\psi_2,a_2)
    \end{aligned}
\right.
\end{equation}
with initial condition
\begin{equation*}
\left\{
    \begin{aligned}
      \psi_2(s,x)&=\psi(x)\\
      a_2(s,k)&=a(k)
    \end{aligned}
\right.\; .
\end{equation*}
We can write the solution of such system as
\begin{equation*}
  \left\{
    \begin{aligned}
      \psi_2(t,x)=&\, U_2^*(t,s)\psi(x) U_2(t,s)\\
      a_2(t,k)=&\, U_2^*(t,s)a(k) U_2(t,s)
    \end{aligned}
\right.
\end{equation*}
and we call $U_2(t,s)$ the evolution of quantum fluctuations. Its precise definition will be given in section~\ref{sec:quant-fluct}. The reasoning above is purely formal, a rigorous proof of the convergence of $W(t,s)$ to $U_2(t,s)$ when $\lambda\to 0$ is needed. We give it in section~\ref{sec:convergence-wt-s}. In order to do that we need to differentiate $W(t,s)$ and $U_2(t,s)$ with respect to $t$ and $s$, and that is not possible on a suitable dense domain of $\mathscr{H}$; however passing to the interaction representation we are able to differentiate. It is then useful to define
\begin{equation}
  \label{eq:14}
    \W(t,s)=U_0^*(t)W(t,s)U_0(s)\; ;\; \U_2(t,s)=U_0^*(t)U_2(t,s) U_0(s)\; .
\end{equation}

\subsection{Definition of spaces $\mathscr{H}^\delta$ and notations about norms.}
\label{sec:defin-spac-mathscrhd}

Let $B\geq 0$ a self-adjoint operator, we define $Q(B)\subseteq \mathscr{H}$ the form domain of $B$, \emph{i.e.} $Q(B)=D(B^{1/2})$. $Q(B)$ is a Hilbert space with norm $\bigl\lVert(B+1)^{1/2}\Phi\bigr\rVert$. We denote $Q^*(B)$ the completion of $\mathscr{H}$ in the norm $\bigl\lVert(B+1)^{-1/2}\Phi\bigr\rVert$.

We can then define spaces $\mathscr{H}^\delta$, $\delta\in\mathds{R}$:
\begin{equation}
  \label{eq:hdelta}
  \mathscr{H}^\delta\equiv\left\{
  \begin{aligned}
    &Q(N^\delta)&\text{ if $\delta\geq 0$}\\
    &Q^*(N^{\ass{\delta}}) &\text{ if $\delta< 0$}
  \end{aligned}
\right .\; .
\end{equation}
Each $\mathscr{H}^\delta$ is a Hilbert space in the norm
\begin{equation*}
  \norm{\Phi}_\delta\equiv\norm{(N+1)^{\delta/2}\Phi}=\norm{(N_1+N_2+1)^{\delta/2}\Phi}\; .
\end{equation*}
We will denote with $\mathscr{B}(\delta';\delta)$ the space of bounded operators from $\mathscr{H}^{\delta'}$ to $\mathscr{H}^\delta$.

The norm of $\mathscr{H}$ is denoted by $\norm{\,\cdot\,}$, the one of $L^2(\mathds{R}^3)$ by $\norm{\,\cdot\,}_2$. The norm of a space $X$ will be denoted explicitly by $\norm{\,\cdot\,;X}$ or $\norm{\,\cdot\,}_X$, with the exception of spaces $\mathscr{H}^\delta$ and $L^p(\mathds{R}^3)$ whose norm will be denoted respectively by $\norm{\,\cdot\,}_\delta$ and $\norm{\,\cdot\,}_p$ (if the context avoids confusion).

This paper is organised as follows: in section~\ref{sec:stat-main-results} we present the main results of this work (theorems~\ref{sec:conv-quant-evol-1},~\ref{sec:class-limit-annih-1} and~\ref{sec:class-limit-annih-2}) and we give a brief summary of the proof strategy; in section~\ref{sec:classical-theory} we analyse the system of classical equations of the theory, and we prove existence and uniqueness of a solution in $\mathscr{C}^0(\mathds{R},L^2(\mathds{R}^3)\times L^2(\mathds{R}^3))$; in section~\ref{sec:quantum-theory} we describe the quantum Hamiltonian and the evolution of quantum fluctuations, as well as proving theorem~\ref{sec:conv-quant-evol-1}; finally section~\ref{sec:class-limit-annih-3} is dedicated to prove theorems~\ref{sec:class-limit-annih-1} and~\ref{sec:class-limit-annih-2}.

\section{Statement of main results and outline of the proofs.}
\label{sec:stat-main-results}

\subsection{Convergence of quantum evolution.}
\label{sec:conv-quant-evol}
As discussed in the introduction, the convergence in a suitable sense of the quantum evolution between coherent states $W(t,s)$ to the evolution of quantum fluctuations $U_2(t,s)$ has to be proved in a rigorous way. The result we can prove is strong convergence of quantum evolution in interaction representation $\W(t,s)$ to the corresponding evolution of fluctuations $\U_2(t,s)$:
\begin{thm}\label{sec:conv-quant-evol-1}
The following strong limit exists in $\mathscr{H}$:
 \begin{equation*}
   \underset{\lambda\to 0}{\mathrm{s-lim}}\:\widetilde{W}(t,s)=\widetilde{U}_2(t,s)\; ,
 \end{equation*}
uniformly in $t,s$ on compact intervals.
\end{thm}
However, since $\widetilde{W}(t,s)$, $\widetilde{U}_2(t,s)$ and $U_0(t)$ are all unitary operators on $\mathscr{H}$, we have also convergence of $W(t,s)$ to $U_2(t,s)$:
\begin{corollary*}
 \begin{equation*}
   \underset{\lambda\to 0}{\mathrm{s-lim}}\:W(t,s)=U_2(t,s)\; ,
 \end{equation*}
uniformly in $t,s$ on compact intervals.
\end{corollary*}

Formally, the generator of $U_2(t,s)$ cancels out the $\lambda$-independent part of the generator of $W(t,s)$, the remaining part converging strongly to zero when $\lambda\to 0$ on a suitable dense subspace of $\mathscr{H}$. So the basic idea of the proof is to use Duhamel's formula to write the difference of the two unitary operators as the integral of the derivative of their product, then use the cancellation between generators to prove strong convergence. The problem is to prove differentiability of both $U_2(t,s)$ and $W(t,s)$ in $t$ (or $s$), since their generators depend on time. This could be done only passing to interaction representation and thus getting rid of the free part $H_0$. The differentiability of $\U_2(t,s)$ is proved in section~\ref{sec:quant-fluct}, introducing a cut off in the total number of particles $N$; the differentiability of $\W(t,s)$ is proved in section~\ref{sec:evolution-wt-s}. Then using $\W(t,s)- \U_2(t,s)$, instead of $W(t,s)-U_2(t,s)$, we are able to write derivatives, and perform the suitable cancellations. The fact that $\U_2(t,s)$ is in $\mathscr{B}(\delta;\delta)$ (proposition~\ref{sec:evol-quant-fluct-6}), with a bound that does not depend on $\lambda$, ensures that everything remains bounded when $\lambda\to 0$. The complete proof can be found in section~\ref{sec:convergence-wt-s}.

\subsection{Classical limit of annihilation and creation operators.}
\label{sec:class-limit-annih}

The classical solution of~\eqref{eq:2} $(u(t),\alpha(t))$ is expected to be the mean field limit correspondent of the quantum variables $(\lambda\psi(t),\lambda a(t))$. This is true if we average the quantum variables over suitable $\lambda$-dependent coherent states:
\begin{thm}\label{sec:class-limit-annih-1}
  Let $\Phi\in \mathscr{H}^\delta$, with $\delta\geq 9$, $(u_0,\alpha_0)\in L^2(\mathds{R}^3)\times L^2(\mathds{R}^3)$. Define
  \begin{equation*}
    (u(t,\cdot),\alpha(t,\cdot))\in \mathscr{C}^0(\mathds{R},L^2(\mathds{R}^3)\times L^2(\mathds{R}^3))
  \end{equation*}
to be the solution of~\eqref{eq:2} with initial conditions $(u_0,\alpha_0)$. Then the statements below are valid:
\begin{enumerate}[i.]
\item The following limits exist in $L^2(\mathds{R}^3)$, when $\lambda\to 0$:
  \begin{align*}
    \langle\lambda\psi(t,\cdot)\rangle_{C\Phi}&\underset{\lambda\to 0}{\overset{L^2(\mathds{R}^3)}{\longrightarrow}}u(t,\cdot)\\
    \langle\lambda\psi^*(t,\cdot)\rangle_{C\Phi}&\underset{\lambda\to 0}{\overset{L^2(\mathds{R}^3)}{\longrightarrow}}\bar{u}(t,\cdot)\\
    \langle\lambda a(t,\cdot)\rangle_{C\Phi}&\underset{\lambda\to 0}{\overset{L^2(\mathds{R}^3)}{\longrightarrow}}\alpha(t,\cdot)\\
    \langle\lambda a^*(t,\cdot)\rangle_{C\Phi}&\underset{\lambda\to 0}{\overset{L^2(\mathds{R}^3)}{\longrightarrow}}\bar{\alpha}(t,\cdot)\; .
  \end{align*}
\item There are two positive constants $K_1$ and $K_2$ such that
  \begin{align*}
    \norm{\langle\lambda\psi(t,\cdot)\rangle_{C\Phi}-u(t,\cdot)}_2&\leq \lambda K_1(1+\ass{t})e^{K_2\ass{t}}\norm{\Phi}_\delta^2\\
    \norm{\langle\lambda\psi^*(t,\cdot)\rangle_{C\Phi}-\bar{u}(t,\cdot)}_2&\leq \lambda K_1(1+\ass{t})e^{K_2\ass{t}}\norm{\Phi}_\delta^2\\
    \norm{\langle\lambda a(t,\cdot)\rangle_{C\Phi}-\alpha(t,\cdot)}_2&\leq \lambda K_1(1+\ass{t})e^{K_2\ass{t}}\norm{\Phi}_\delta^2\\
    \norm{\langle\lambda a^*(t,\cdot)\rangle_{C\Phi}-\bar{\alpha}(t,\cdot)}_2&\leq \lambda K_1(1+\ass{t})e^{K_2\ass{t}}\norm{\Phi}_\delta^2\; .
  \end{align*}
\item If $\Phi=\Omega$, the vacuum state of $\mathscr{H}$, then
  \begin{align*}
    \norm{\langle\lambda\psi(t,\cdot)\rangle_{C\Omega}-u(t,\cdot)}_2&\leq \lambda^2 K_1\ass{t}e^{K_2\ass{t}}\\
    \norm{\langle\lambda\psi^*(t,\cdot)\rangle_{C\Omega}-\bar{u}(t,\cdot)}_2&\leq \lambda^2 K_1\ass{t}e^{K_2\ass{t}}\\
    \norm{\langle\lambda a(t,\cdot)\rangle_{C\Omega}-\alpha(t,\cdot)}_2&\leq \lambda^2 K_1\ass{t}e^{K_2\ass{t}}\\
    \norm{\langle\lambda a^*(t,\cdot)\rangle_{C\Omega}-\bar{\alpha}(t,\cdot)}_2&\leq \lambda^2 K_1\ass{t}e^{K_2\ass{t}}\; .
  \end{align*}
\end{enumerate}
\end{thm}
The basic tool we need to prove the theorem is a bound of $\norm{\W(t,s)\Phi}_\delta$ that is convergent when $\lambda\to 0$. As discussed at the beginning of section~\ref{sec:class-limit-annih-3}, where such convergent bound is proved, we need to perform a regularisation in both numbers $N_1$ and $N_2$ of particles. Then comparing the regularised operator $\W_\nu(t,s)$ with $\W(t,s)$ we obtain the bound with suitable $\lambda$-dependence (proposition~\ref{sec:useful-estimates-3}). The price we have to pay is that the bound holds only on a subspace with much more regularity than a priori expected. Once we have this bound, the proof of the first two points of the theorem is a direct consequence of it, and can be found in lemma~\ref{sec:gener-case:-conv-1}. To improve the rate of convergence to $\lambda^2$, as we are able to do in the last point, we compare the quantum dynamics $\W(t,s)$ with the dynamics of fluctuations $\U_2(t,s)$. The fact that, although $\U_2$ does not preserve the number of particles, the $\U_2$-evolved quantum fields applied to the vacuum yield still one-particle states (proposition~\ref{sec:useful-estimates-4}) leads to the cancellation of the leading term of order $\lambda$, improving thus the rate of convergence to $\lambda^2$. The complete proof can be found in lemma~\ref{sec:proof-theor-refs}.

The results of the theorem above can be extended to the averages of normal ordered products of quantum variables, and to states with fixed number of particles as well as coherent states:
\begin{thm}\label{sec:class-limit-annih-2}
Let $u_0,\alpha_0\in L^2(\mathds{R}^3)$ such that $\norm{u_0}_2=\norm{\alpha_0}_2=1$. Define
  \begin{equation*}
    (u(t,\cdot),\alpha(t,\cdot))\in \mathscr{C}^0(\mathds{R},L^2(\mathds{R}^3)\times L^2(\mathds{R}^3))
  \end{equation*}
to be the solution of~\eqref{eq:2} with initial conditions $(u_0,\alpha_0)$,
\begin{equation*}
      (u_\theta(t,\cdot),\alpha_\theta(t,\cdot))\in \mathscr{C}^0(\mathds{R},L^2(\mathds{R}^3)\times L^2(\mathds{R}^3))
\end{equation*}
the solution of~\eqref{eq:2} with initial conditions $(u_0,\alpha_0(\theta))$, $\alpha_0(\theta)\equiv\exp\{-i\theta\}\alpha_0$ for all $\theta\in\mathds{R}$. Then the statements below are valid for all $q$ ,$r$, $h$, $l\in\mathds{N}$, $\delta=q+r+h+l$:
\begin{enumerate}[i.]
\item The following limits exist in $L^2(\mathds{R}^{3\delta})$ when $\lambda\to 0$:
  \begin{align*}
\langle \psi^*(q)\psi(r)a^*(h)a(l) \rangle_\Lambda(t)&\underset{\lambda\to 0}{\overset{L^2(\mathds{R}^{3\delta})}{\longrightarrow}}\bar{u}^{\otimes_q}u^{\otimes_r}\bar{\alpha}^{\otimes_h}\alpha^{\otimes_l}(t)\\
\langle \psi^*(q)\psi(r)a^*(h)a(l) \rangle_\Psi(t)&\underset{\lambda\to 0}{\overset{L^2(\mathds{R}^{3\delta})}{\longrightarrow}}\delta_{qr}\bar{u}^{\otimes_q}u^{\otimes_r}\bar{\alpha}^{\otimes_h}\alpha^{\otimes_l}(t)\\
\langle \psi^*(q)\psi(r)a^*(h)a(l) \rangle_\Theta(t)&\underset{\lambda\to 0}{\overset{L^2(\mathds{R}^{3\delta})}{\longrightarrow}}\delta_{qr}\int_0^{2\pi}\frac{\de\theta}{2\pi}\bar{u}_\theta^{\otimes_q}u_\theta^{\otimes_r}\bar{\alpha}_\theta^{\otimes_h}\alpha_\theta^{\otimes_l}(t)\; ,
  \end{align*}
$\delta_{qr}$ being the function equal to $1$ when $q=r$, $0$ otherwise.
\item For all $\Xi\in\{\Lambda,\Psi,\Theta\}$ there are two positive constants $K_1(\Xi)$ and $K_2(\Xi)$ that depend on $p,q,h,l$ such that
  \begin{equation*}
\bigl\lVert\langle \psi^*(q)\psi(r)a^*(h)a(l) \rangle_\Lambda(t)-\bar{u}^{\otimes_q}u^{\otimes_r}\bar{\alpha}^{\otimes_h}\alpha^{\otimes_l}(t)\bigr\rVert_{L^2(\mathds{R}^{3\delta})}\leq \lambda^2 K_1(\Lambda)\ass{t}e^{K_2(\Lambda)\ass{t}}    
  \end{equation*}
  \begin{equation*}
\bigl\lVert\langle \psi^*(q)\psi(r)a^*(h)a(l) \rangle_\Psi(t)-\delta_{qr}\bar{u}^{\otimes_q}u^{\otimes_r}\bar{\alpha}^{\otimes_h}\alpha^{\otimes_l}(t)\bigr\rVert_{L^2(\mathds{R}^{3\delta})}\leq \delta_{qr}\lambda^2 K_1(\Psi)\ass{t}e^{K_2(\Psi)\ass{t}}    
  \end{equation*}
  \begin{equation*}
\bigl\lVert\langle \psi^*(q)\psi(r)a^*(h)a(l) \rangle_\Theta(t)-\delta_{qr}\int_0^{2\pi}\frac{\de\theta}{2\pi}\bar{u}_\theta^{\otimes_q}u_\theta^{\otimes_r}\bar{\alpha}_\theta^{\otimes_h}\alpha_\theta^{\otimes_l}(t)\bigr\rVert_{L^2(\mathds{R}^{3\delta})}\leq \delta_{qr}\lambda^2 K_1(\Theta)\ass{t}e^{K_2(\Theta)\ass{t}}
  \end{equation*}
\end{enumerate}  
\end{thm}
The proof is carried out in the same way as in the last point of theorem~\ref{sec:class-limit-annih-1}: we use the results proved by~\citet{2011arXiv1103.0948C} and~\citet{MR2530155} (lemma~\ref{sec:trans-ampl-8}) to write fixed particles states as suitable combinations of coherent states, and then use the convergent bound of $\norm{\W(t,s)\Phi}_\delta$; the comparison with $\U_2(t,s)$ dynamics let us improve the rate of convergence to $\lambda^2$. The proof for $\Theta$ vectors (fixed number of both non-relativistic and relativistic particles) can be found in section~\ref{sec:proof-th}, the other two cases being analogous.

\subsection{Discussion of the results.}
\label{sec:discussion-results}

As expected we are able to put on solid mathematical foundations the results of convergence described naïvely in the introduction. The solution of coupled Schrödinger and Klein-Gordon equations is the classical counterpart of the quantum variables of the system (namely annihilation and creation operators). This is set in theorem~\ref{sec:class-limit-annih-1}, and the convergence of the latter to the former is intended to be the convergence in $L^2(\mathds{R}^3)$ of transition amplitudes between coherent states of quantum operators to classical functions. The dynamics of quantum fluctuations is governed by the linearization of Heisenberg equations around the classical solution: this is proved by the convergence stated in theorem~\ref{sec:conv-quant-evol-1}, keeping in mind the differential properties of $\U_2(t,s)$ (proposition~\ref{sec:evol-quant-fluct-6}). In theorem~\ref{sec:class-limit-annih-2} we extend the convergence of transition amplitudes to normal ordered products of creation and annihilation operators. The behaviour of $U_2(t,s)$ evolution of one-particle states enables to improve the rate of convergence of such amplitudes, between suitable states ($\Lambda$, $\Psi$ and $\Theta$), to order $\lambda^2$ instead of $\lambda$.

The result about fixed particles $\Theta$-vectors stated in theorem~\ref{sec:class-limit-annih-2} deserves a specific comment. The classical limit in this case differs from the expected product of classical solutions. A phase multiplying the initial relativistic datum arises, and the limit is an average of the product of classical solutions corresponding to such varying initial data. Observe that quantum dynamics of the relativistic field does not preserve the number of particles, while non-relativistic particle number is preserved by quantum evolution. So initial states with fixed number of relativistic particles could be seen as a bad choice to describe the theory. The non-classical residue obtained in the limit is possibly related to this problem and should emerge also in the classical limit of other systems that does not preserve the number of particles.

\section{Classical theory.}
\label{sec:classical-theory}

In this section we prove existence and uniqueness of a solution in $\mathscr{C}^0(\mathds{R},L^2(\mathds{R}^3)\times L^2(\mathds{R}^3))$ of the classical system~\eqref{eq:2} with initial data in $L^2(\mathds{R}^3)$ (proposition~\ref{cl.p:1}); furthermore we prove such solution is continuous in $\mathscr{C}^0(\mathds{R},L^2(\mathds{R}^3)\times L^2(\mathds{R}^3))$ with respect to a $L^2(\mathds{R}^3)$-continuous variation of the initial $\alpha$-datum (lemma~\ref{cl.rem:2}: this result is needed in theorem~\ref{sec:class-limit-annih-2} for $\Theta$ vectors, so the integration of classical solutions corresponding to different initial data makes sense). 

Let $\alpha_0, u_0\in L^2(\mathds{R}^3)$, and define $U_{01}(t)\equiv \exp(i\Delta t/2)$, $U_{02}(t)\equiv \exp(-i\omega t)$, with
\begin{equation*}
  (\omega^\lambda \alpha)(x)=(2\pi)^{-3/2}\int\ide{\xi}e^{i\xi x}(\mu^2+\ass{\xi}^2)^{\lambda/2}\mathcal{F}(\alpha)(\xi)\text{ , $\mu\geq 0$.}
\end{equation*}
We consider the following system of integral equations:
\begin{equation}
\label{eq:15}
 \left\{
   \begin{aligned}
     u(t)&=U_{01}(t)u_0-i(2\pi)^{-3/2}\int_0^t\ide{\tau}U_{01}(t-\tau) u(\tau)(\mathcal{F}^{-1}(\chi) * A(\tau))\\
     \alpha(t)&=U_{02}(t)\alpha_0-i\frac{(2\pi)^{-3/2}}{\sqrt{2}}\int_0^t\ide{\tau}U_{02}(t-\tau)\omega^{-1/2}(\mathcal{F}^{-1}(\chi)*\ass{u(\tau)}^2)
   \end{aligned}
\right. 
\end{equation}
where $A(t)=\omega^{-1/2}(\alpha(t)+\bar{\alpha}(t))$. We want to prove the existence of a unique solution of the system in $\mathscr{C}^0(\mathds{R},L^2(\mathds{R}^3)\times L^2(\mathds{R}^3))$. If $(u,\alpha)$ is such a solution, then $(u,\mathcal{F}(\alpha))$ is the $\mathscr{C}^0(\mathds{R},L^2(\mathds{R}^3)\times L^2(\mathds{R}^3))$ solution of~\eqref{eq:2} with initial data $(u_0,\mathcal{F}(\alpha_0))$.

\subsection{Existence and uniqueness of the solution.}
\label{sec:exist-uniq-solut}

\begin{lemma}
  \label{sec:exist-uniq-solut-1}
Let $V\in \mathscr{C}^0(\mathds{R},L^\infty(\mathds{R}^3))$. Then, $\forall u_0\in L^2(\mathds{R}^3)$, $\exists !u\in \mathscr{C}^0(\mathds{R},L^2(\mathds{R}^3))$ solution of
\begin{equation}\label{eq:16}
  u(t)=U_{01}(t)u_0-i\int_0^t\ide{\tau}U_{01}(t-\tau) V(\tau)u(\tau)\; .
\end{equation}
Furthermore if
\begin{equation*}
  u_j(t)=U_{01}(t)u_0-i\int_0^t\ide{\tau}U_{01}(t-\tau) V_j(\tau)u_j(\tau)\text{ with $j=1,2$}
\end{equation*}
we have the following estimate:
\begin{equation}\label{eq:17}
\begin{split}
  \norm{u_1(t)-u_2(t)}_2\leq \norm{u_2;\mathscr{C}^0([0,t],L^2)}\int_0^t\ide{\tau}\norm{(V_1-V_2)(\tau);L^\infty(\mathds{R}^3)}\\
\exp\biggl(\ass{\int_0^t\ide{\tau}\norm{V_1(\tau);L^\infty(\mathds{R}^3)}}\biggr)\; .
\end{split}\end{equation}
Finally if $V$ is real then $\norm{u(t)}_2=\norm{u_0}_2$ for all $t$ (the charge is conserved). 
\end{lemma}
With the aid of this lemma we can prove the existence of a unique solution of the system~\eqref{eq:15}, this is done in the following proposition:
\begin{proposition}
  \label{cl.p:1}
Let $u_0,\alpha_0\in L^2(\mathds{R}^3)$. Then $\exists ! (u(\cdot),\alpha(\cdot))$ in $\mathscr{C}^0(\mathds{R},L^2(\mathds{R}^3)\times L^2(\mathds{R}^3))$ solution of the integral system~\eqref{eq:15}.
\end{proposition}
\begin{proof}
  For all $j=1,2,\dotsc$ the systems:
  \begin{equation}
    \label{cl.eq:5}
     \left\{
   \begin{aligned}
     u_j(t)&=u_0(t)-i(2\pi)^{-3/2}\int_0^t\ide{\tau}U_{01}(t-\tau) u_j(\tau)(\mathcal{F}^{-1}(\chi) * A_{j-1}(\tau))\\
     \alpha_j(t)&=\alpha_0(t)-i\frac{(2\pi)^{-3/2}}{\sqrt{2}}\int_0^t\ide{\tau}U_{02}(t-\tau)\omega^{-1/2}(\mathcal{F}^{-1}(\chi)*\ass{u_{j-1}(\tau)}^2)
   \end{aligned}
\right. 
  \end{equation}
with $u_0(t)\equiv U_{01}(t)u_0$ and $\alpha_0(t)\equiv U_{02}(t)\alpha_0$, have a unique solution $(u_j,\alpha_j)\in \mathscr{C}^0(\mathds{R},L^2(\mathds{R}^3)\otimes L^2(\mathds{R}^3))$ by lemma~\ref{sec:exist-uniq-solut-1}. Now let $t\in I=[0,\epsilon]$, and define the map $S$ on $\mathscr{C}^0(I,L^2(\mathds{R}^3)\otimes L^2(\mathds{R}^3))$ as
\begin{equation*}
  \begin{split}
    S\left(
    \begin{gathered}
      u(t)\\\alpha(t)
    \end{gathered}\right)=
  \left(\begin{gathered}
u_0(t)-i(2\pi)^{-3/2}\int_0^t\ide{\tau}U_{01}(t-\tau) u(\tau)(\mathcal{F}^{-1}(\chi) * A(\tau))\\
\alpha_0(t)-i\frac{(2\pi)^{-3/2}}{\sqrt{2}}\int_0^t\ide{\tau}U_{02}(t-\tau)\omega^{-1/2}(\mathcal{F}^{-1}(\chi)*\ass{u(\tau)}^2)
  \end{gathered}\right)\; .
  \end{split}
\end{equation*}
Let $u_1,u_2,\alpha_1,\alpha_2\in \mathscr{C}^0(I,L^2(\mathds{R}^3))$, and define
\begin{equation*}
  S\left(\begin{gathered}
    u_1(t)\\\alpha_1(t)
  \end{gathered}\right)-S\left(\begin{gathered}
    u_2(t)\\\alpha_2(t)
  \end{gathered}\right)=\left(\begin{gathered}
    u'_1(t)\\\alpha'_1(t)
  \end{gathered}\right)-\left(\begin{gathered}
    u'_2(t)\\\alpha'_2(t)
  \end{gathered}\right)\; .
\end{equation*}
By estimate~\eqref{eq:17} of lemma~\ref{sec:exist-uniq-solut-1} and conservation of charge we obtain:
\begin{equation*}
  \begin{split}
 \sup_{t\in I} \norm{u'_1(t)-u'_2(t)}_2
\leq C_s (2\pi)^{-3/2}\epsilon\norm{\mathcal{F}^{-1}(\chi)}_{3/2} \exp\Bigl(C_s (2\pi)^{-3/2}\epsilon \norm{\mathcal{F}^{-1}(\chi)}_{3/2}\\\max_{j=1,2}\norm{\alpha_j;\mathscr{C}^0(I,L^2)}\Bigr)\max_{j=1,2}\norm{u_j;\mathscr{C}^0(I,L^2)}\norm{\alpha_1-\alpha_2;\mathscr{C}^0(I,L^2)}
  \end{split}
\end{equation*}
\begin{equation*}
  \begin{split}
    \sup_{t\in I} \norm{\alpha'_1(t)-\alpha'_2(t)}_2
\leq \sqrt{2}(2\pi)^{-3/2}\epsilon\norm{\mathcal{F}^{-1}(\chi)}_2\max_{j=1,2}\norm{u_j;\mathscr{C}^0(I,L^2)}\norm{u_1-u_2;\mathscr{C}^0(I,L^2)}\; .
  \end{split}
\end{equation*}
Choosing $\epsilon$ small enough $S$ becomes a strict contraction. Using conservation of charge the solution is extended to all $\mathscr{C}^0(\mathds{R},L^2(\mathds{R}^3)\times L^2(\mathds{R}^3))$.
\end{proof}

\subsection{Interaction respresentation and continuity with respect to initial data.}
\label{sec:inter-respr-cont}

We formulate a couple of useful lemmas whose proof is straightforward.
\begin{lemma}\label{cl.rem:1}
  If $(u(t),\alpha(t))$ is the solution of~\eqref{eq:15} in $\mathscr{C}^0(\mathds{R},L^2(\mathds{R}^3)\times L^2(\mathds{R}^3))$, define $(\widetilde{u}(t),\widetilde{\alpha}(t))\equiv (U_{01}(-t)u(t),U_{02}(-t)\alpha(t))$.

Then $(\widetilde{u}(t),\widetilde{\alpha}(t))\in\mathscr{C}^1(\mathds{R},L^2(\mathds{R}^3)\times L^2(\mathds{R}^3))$ and we have that:
\begin{align*}
  i\partial_t \widetilde{u}(t)&=(2\pi)^{-3/2}U_{01}(-t)\Bigl(\mathcal{F}^{-1}(\chi) * A(t)\Bigr)u(t)\\
  i\partial_t \widetilde{\alpha}(t)&=\frac{(2\pi)^{-3/2}}{\sqrt{2}}U_{02}(-t)\Bigl(\omega^{-1/2}(\mathcal{F}^{-1}(\chi)*\ass{u(t)}^2)\Bigr)\; .
\end{align*}
\end{lemma}
\begin{lemma}\label{cl.rem:2}
  Let $(u_1(\cdot),\alpha_1(\cdot))$ and $(u_2(\cdot),\alpha_2(\cdot))$ be the solutions of~\eqref{eq:15} in $\mathscr{C}^0(\mathds{R},L^2\times L^2)$ corresponding respectively to initial data $(u_0,\alpha_{01})$ and $(u_0,\alpha_{02})$ both in $L^2\times L^2$. Then if $\alpha_{01}\to_{L^2}\alpha_{02}$, then $(u_1(\cdot),\alpha_1(\cdot))\to (u_2(\cdot),\alpha_2(\cdot))$ in $\mathscr{C}^0(\mathds{R},L^2(\mathds{R}^3)\times L^2(\mathds{R}^3))$.
\end{lemma}
\section{Quantum Theory.}
\label{sec:quantum-theory}

This section is devoted to the study of the quantum Hamiltonian and its corresponding evolution, and to define the evolution of fluctuations. The self-adjointness of Nelson Hamiltonian is discussed in section~\ref{sec:self-adjointness-h}; the evolution of quantum fluctuations $\U_2(t,s)$ is defined in section~\ref{sec:quant-fluct}, and its most important properties (of which we make extensive use throughout the rest of the work) are stated in proposition~\ref{sec:evol-quant-fluct-6}; quantum evolution between coherent states $\W(t,s)$ and its differential properties are discussed in section~\ref{sec:evolution-wt-s}; finally in section~\ref{sec:convergence-wt-s} theorem~\ref{sec:conv-quant-evol-1} is proved.

Let $f\in L^2(\mathds{R}^3)$, we recall the definition of the annihilation and creation operators of $\mathscr{H}$ given in section~\ref{sec:nels-hamilt-fock}:
\begin{equation*}
    (\psi(f)\Phi)_{p,n}(x_1,\dotsc,x_p;k_1,\dotsc,k_n)=\sqrt{p+1}\int\ide{x}f(x)\Phi_{p+1,n}(x,x_1,\dotsc,x_p;k_1,\dotsc,k_n)\mspace{13mu}
\end{equation*}
\begin{equation*}
    (\psi^*(f)\Phi)_{p,n}(x_1,\dotsc,x_p;k_1,\dotsc,k_n)=\frac{1}{\sqrt{p}}\sum_{i=1}^p f(x_i)\Phi_{p-1,n}(x_1,\dotsc,\hat{x}_i,\dotsc,,x_p;k_1,\dotsc,k_n)\mspace{8mu}
\end{equation*}
\begin{equation*}
  (a(f)\Phi)_{p,n}(x_1,\dotsc,x_p;k_1,\dotsc,k_n)=\sqrt{n+1}\int\ide{k}f(k)\Phi_{p,n+1}(x_1,\dotsc,x_p;k,k_1,\dotsc,k_n)
\end{equation*}
\begin{equation*}
   (a^*(f)\Phi)_{p,n}(x_1,\dotsc,x_p;k_1,\dotsc,k_n)=\frac{1}{\sqrt{n}}\sum_{j=1}^nf(k_j)\Phi_{p,n-1}(x_1,\dotsc,x_p;k_1,\dotsc,\hat{k}_j,\dotsc,k_n)\; .\mspace{3.5mu}
\end{equation*}
On $\mathscr{H}_p$ it is useful to define slightly different relativistic annihilation and creation operators; let $f\in L^\infty(\mathds{R}^{3p},L^2(\mathds{R}^3))$, $\Phi\in\mathscr{H}_p$ and define
\begin{align*}
  (a(f)\Phi)_{p,n}&=\sqrt{n+1}\int\ide{k}f(x_1,\dotsc,x_p;k)\Phi_{p,n+1}(x_1,\dotsc,x_p;k,k_1,\dotsc,k_n)\\
  (a^*(f)\Phi)_{p,n}&=\frac{1}{\sqrt{n}}\sum_{j=1}^nf(x_1,\dotsc,x_p;k_j)\Phi_{p,n-1}(x_1,\dotsc,x_p;k_1,\dotsc,\hat{k}_j,\dotsc,k_n)\; ;
\end{align*}
it will be clear from the context what type of operators we use. From now on set
\begin{equation*}
    f=\sum_{j=1}^p f_j\; ,\; f_j=\lambda f_0e^{-ik\cdot x_j}\; ,\; f_0=(2\pi)^{-3/2}(2\omega)^{-1/2}\chi_\sigma\; ;
\end{equation*}
we remark that for all $\sigma\in \mathds{R}$, $f_0\in L^2(\mathds{R}^3)$ with $\omega^\delta f_0\in L^2(\mathds{R}^3)$ for all $\delta\geq -1/2$, even when $\mu=0$. Then on $\mathscr{H}_p$ we can write $H_{I}\bigr\rvert_p= a(\bar{f})+a^*(f)$. To be precise, we define $H=(H_0+H_{I})^{\overline{\phantom{=}}}$ rather than $H=H_0+H_I$ as we did in section~\ref{sec:nels-hamilt-fock}. The following estimates are useful to prove self-adjointness of $H$, their proof is standard~\citep[see][as a reference]{MR2205462}.
\begin{lemma}
  \label{sec:quantum-theory-6}
Let $g\in L^\infty(\mathds{R}^{3p},L^2(\mathds{R}^3))$ such that also $\omega^{-1/2}(k)g(x_1,\dotsc,x_p;k)$ is in the same space. Then, for all $\Phi\in D(H_{02}^{1/2})\cap \mathscr{H}_p$, intended as the domain on which the RHS are finite, the following estimates hold:
\begin{align*}
  \norm{a(g)\Phi}^2&\leq\norm{\omega^{-1/2}g}_*^2\norm{H_{02}^{1/2}\Phi}^2\; ;\; \norm{a^*(g)\Phi}^2\leq\norm{\omega^{-1/2}g}_*^2\norm{H_{02}^{1/2}\Phi}^2+\norm{g}_*^2\norm{\Phi}^2\; ;
\end{align*}
where $\norm{\cdot}_*$ is the $L^\infty(\mathds{R}^{3p},L^2(\mathds{R}^3))$-norm.

Let now $g\in L^\infty(\mathds{R}^{3p},L^2(\mathds{R}^3))$, and $\Phi\in D(N_2^{1/2})\cap\mathscr{H}_p$, then:
\begin{align*}
  \norm{a(g)\Phi}&\leq\norm{g}_*\norm{N_2^{1/2}\Phi}\; ;\; \norm{a^*(g)\Phi}\leq\norm{g}_*\norm{(N_2+1)^{1/2}\Phi}\; .
\end{align*}
\end{lemma}
\begin{corollary*}
  For all $\Phi\in D(N_1^2+N_2)\cap D(N_1N_2^{1/2})$ we have that:
  \begin{equation*}
    \norm{H_I\Phi}\leq 2\lambda\norm{f_0}_2\norm{N_1(N_2+1)^{1/2}\Phi}\leq \lambda\norm{f_0}_2\norm{(N_1^2+N_2+1)\Phi}\; .
  \end{equation*}
\end{corollary*}

\subsection{Self-adjointness of $H$.}
\label{sec:self-adjointness-h}

Let $\Phi_p$, $H_0\bigr\rvert_p$ and $H\bigr\rvert_p$ be the projections of $\Phi\in\mathscr{H}$, $H_0$ and $H$ respectively on $\mathscr{H}_p$.
\begin{proposition}\label{sec:quantum-theory-2}
$\phantom{i}$
\begin{enumerate}[i.]
\item\label{qt.item:20} $H\bigr\rvert_p$ is self-adjoint on $\mathscr{H}_p$ with domain $D(H_0\bigr\rvert_p)$.
\item\label{qt.item:21} $H$ is self-adjoint on $\mathscr{H}$ with domain $D(H)$ defined as following:
  \begin{equation*}
    D(H)=\Bigl\{\Phi\in\mathscr{H}:\; \sum_{p=0}^\infty\norm{H\bigr\rvert_p\Phi_p}^2<\infty,\; \Phi_p\in D(H_0\bigr\rvert_p)\Bigr\}\; .
  \end{equation*}
\item\label{qt.item:22} On $\mathscr{H}$, we have the following inclusions:
  \begin{align*}
    D(H_0)&\supseteq D(H)\cap D(N_1^2+N_2)\; ;\; D(H)\supseteq D(H_0)\cap D(N_1^2+N_2)\; .
  \end{align*}
\end{enumerate}
\end{proposition}
\begin{proof}
\emph{\ref{qt.item:20}}. $H_{I}\bigr\rvert_p$ is a Kato perturbation of $H_0\bigr\rvert_p$: by lemma~\ref{sec:quantum-theory-6} we obtain for all $\Phi_p\in D(H_0\bigr\rvert_p)$
\begin{equation*}
  \norm{H_{I}\Phi_p}^2\leq 4\lambda^2\norm{\omega^{-1/2}f_0}^2_2\norm{N_1H_{02}^{1/2}\Phi_p}^2+2\lambda^2\norm{f_0}_2^2\norm{N_1\Phi_p}^2\; .
\end{equation*}
Then for all $\epsilon>0$: $\norm{H_{I}\Phi_p}^2\leq \epsilon^2\norm{H_{02}\Phi_p}^2+\frac{4\lambda^4}{\epsilon^2}\norm{\omega^{-1/2}f_0}_2^4\norm{N_1^2\Phi_p}^2+2\lambda^2\norm{f_0}_2^2\norm{N_1\Phi_p}^2$.

\emph{\ref{qt.item:21}}. Since $H\bigr\rvert_p$ is self-adjoint on $\mathscr{H}_p$ we can define the self-adjoint operator $H$ on $\mathscr{H}$ as a direct sum.

\emph{\ref{qt.item:22}}. To prove the first relation we proceed as following: from the fact that $H_0=H-H_I$ we can write $\forall\Phi\in D(H)\cap D(N_1^2+N_2)$, and a suitable $L>0$
\begin{equation*}
  \norm{H_0\Phi}\leq \norm{H\Phi}+\norm{H_I\Phi}\leq \norm{H\Phi}+L\norm{(N_1^2+N_2+1)\Phi}\; .
\end{equation*}
The second relation is proved in analogous fashion, writing $H=H_0+H_I$.
\end{proof}

\subsection{Invariance of domains.}
\label{sec:invariance-domains}

\begin{lemma}
  \label{sec:quantum-theory-5}
Let $F(\lambda)$ be the spectral family of the operator $N\equiv N_1+N_2$, $\xi(N_1,N_2)$ any $F$-measurable operator-valued function, with domain $D(\xi)$; consider now the operator
\begin{equation*}\begin{split}
  B=\int\de X_q\de Y_r\de K_h\ide{M_l} \bar{g}(x_1,\dotsc,x_q;y_1,\dotsc,y_r;k_1,\dotsc,k_h;m_1,\dotsc,m_l)\\
\psi^*(X_q)\psi(Y_r)a^*(K_h)a(M_l)\; ,
\end{split}\end{equation*}
defined on $D(B)$, with $q,r,h,l\in\mathds{N}$, $q+r+h+l=\delta$. Then:

\begin{enumerate}[i.]
\item\label{qt.item:9} $\xi(N_1,N_2)B\Psi=B \xi(N_1+q-r,N_2+h-l)\Psi$, for suitable $\Psi$.
\item\label{qt.item:10} For all $g\in L^2(\mathds{R}^{3\delta})$ and $\Phi\in D(N^\delta)$ the following estimate holds:
  \begin{align*}
\norm{B\Phi}\leq \norm{g}_{L^2(\mathds{R}^{3\delta})} &\biggl\lVert\frac{\sqrt{N_1!(N_1+q-r)!N_2!(N_2+h-l)!}}{(N_1-r)!(N_2-l)!}\theta(N_1-r)\theta(N_2-l)\Phi\biggr\rVert
 \end{align*}
where $\theta(b)=1$ if $b\geq 0$ and zero otherwise, with $b\in \mathds{Z}$.
\end{enumerate}
\end{lemma}
This lemma is proved by direct calculation and using standard estimates of creation and annihilation operators.
\begin{proposition}
  \label{sec:quantum-theory-3}
Let $\xi$ be any $F$-measurable function, as in the lemma above, then:
\begin{enumerate}[i.]
\item\label{qt.item:1} $U_0(t)\Phi\in D(H_0)$ for all $t\in\mathds{R}$, $\Phi\in D(H_0)$, and $\norm{H_0U_0(t)\Phi}=\norm{H_0\Phi}$.
\item\label{qt.item:2} $U_0(t)\Phi\in D(\xi(N_1,N_2))$ for all $t\in\mathds{R}$, $\Phi\in D(\xi(N_1,N_2))$, and $\norm{\xi(N_1,N_2)U_0(t)\Phi}=\norm{\xi(N_1,N_2)\Phi}$.
\item\label{qt.item:3} $U(t)\Phi\in D(H)$ for all $t\in\mathds{R}$, $\Phi\in D(H)$, and $\norm{HU(t)\Phi}=\norm{H\Phi}$.
\item\label{qt.item:4} $U(t)\Phi\in D((N_1^2+N_2)^\delta)$ for all $t\in\mathds{R}$, $\Phi\in D((N_1^2+N_2)^\delta)$, $\delta\in\mathds{R}$; and
  \begin{equation*}
    \norm{(N_1^2+N_2+1)^\delta U(t)\Phi}\leq \exp(\ass{\delta} c_1(\delta) \lambda\norm{f_0}_2 \ass{t})\norm{(N_1^2+N_2+1)^\delta\Phi}\; ,
  \end{equation*}
with $c_1(\delta)=\max(3,1+2^{\ass{\delta}})$.
\end{enumerate}
\end{proposition}
\begin{proof}
The first three statements are a direct consequence of Stone's theorem and the fact that $H_0$, $N_1$ and $N_2$ commute.

\emph{\ref{qt.item:4}}. Let $\Phi\in D(H_0\bigr\rvert_p)$, $0<h(N_2)$ a bounded operator on $\mathscr{H}_p$ such that $\ran h(N_2)\subset D(N_2^{1/2})$. Define the differentiable quantity
\begin{equation*}
  M(t)\equiv \frac{1}{2}\norm{h(N_2)U(t)\Phi}^2\; .
\end{equation*}
With a bit of manipulation and since $H_0$ commutes with $N_2$ we obtain
\begin{equation*}
  \begin{split}
    \ass{\dtot{}{t}M(t)}\leq 2p\lambda\norm{f_0}_2\Bigl[\nnorm{\sqrt{N_2}\bigl(h(N_2-1)h(N_2)^{-1}-1\bigr)}+\nnorm{\sqrt{N_2}\bigl(h(N_2)h(N_2-1)^{-1}-1\bigr)}\Bigr]\\M(t)
  \end{split}
\end{equation*}
where $\nnorm{\, \cdot\,}$ is the norm of $\mathscr{B}(\mathscr{H},\mathscr{H})$. Let $h\in\mathscr{C}^1$, $h(\cdot)$ and $\ass{h'(\cdot)}$ non-increasing; then
\begin{equation*}
  \begin{split}
   K\equiv \Bigl[\dotsc\Bigr]\leq \sup_{n=0,1,\dotsc}\sqrt{n}\ass{h'(n-1)}h^{-1}(n)+\sup_{n=0,1,\dotsc}\sqrt{n}\ass{h'(n-1)}h^{-1}(n-1)\; .
  \end{split}
\end{equation*}
We are interested in the case $h(n)=(n+j+1)^{-\delta}$, with $\delta\geq 1/2$ (so $\ran h(N_2)\subset D(N_2^{1/2})$) and $j\geq 1$. $h$ satisfies the hypothesis above and $h'(n)=-\delta(n+j+1)^{-\delta-1}$. So we have that
\begin{align*}
  \ass{h'(n-1)}h^{-1}(n)&=\delta(n+j)^{-1}\Bigl(1+\frac{1}{n+j}\Bigr)^\delta\leq \delta 2^\delta (n+j)^{-1}\; ,\\
  \ass{h'(n-1)}h^{-1}(n-1)&=\delta(n+j)^{-1}\; .
\end{align*}
The function $g(x)=\sqrt{x}/(x+j)$, with $x\geq 0$ has a maximum when $x=j$, so $K\leq \frac{1}{2}\delta(1+2^\delta)j^{-1/2}$. We have then the following differential inequality for $M(t)$:
\begin{equation*}
  \ass{\dtot{}{t}M(t)}\leq pj^{-1/2}\lambda\norm{f_0}_2\delta(1+2^\delta)M(t)\; ,
\end{equation*}
so Gronwall's Lemma implies
\begin{equation*}
  M(t)\leq e^{pj^{-1/2}\delta(1+2^\delta)\lambda\norm{f_0}_2 t}M(0)\; .
\end{equation*}
Set now $j=p^2$, with $p\geq 1$:
\begin{equation}\label{qt.eq:8}
  \norm{(N_2+p^2+1)^{-\delta}U(t)\Phi}\leq e^{\delta(1+2^\delta)\lambda\norm{f_0}_2 t}\norm{(N_2+p^2+1)^{-\delta}\Phi}\; ;
\end{equation}
for all $\delta\geq 1/2$ and $\Phi\in D(H_0\bigr\rvert_p)$. Interpolating between $\delta=0$ and $\delta=1$ we extend the result to $0\leq\delta\leq 1$:
\begin{equation}
  \label{qt.eq:9}
    \norm{(N_2+p^2+1)^{-\delta}U(t)\Phi}\leq e^{3\delta\lambda\norm{f_0}_2 t}\norm{(N_2+p^2+1)^{-\delta}\Phi}\; .
\end{equation}
These results extend immediately to all $\Phi\in\mathscr{H}_p$. By duality the bound holds also for $\delta<0$. The result on $\mathscr{H}$ follows by taking the direct sum of all $\mathscr{H}_p$.
\end{proof}

\subsection{Weyl operators.}
\label{sec:weyl-operators}

Weyl operators have been introduced in section~\ref{sec:classical-limit}, here we state the properties we use the most as a proposition~\citep[refer to][for additional information and proofs]{MR530915,MR2205462}.
\begin{proposition}
  \label{sec:quantum-theory-8}
\begin{enumerate}[i.]
    \item\label{qt.item:5} $C(u,\alpha)$ is unitary and strongly continuous as a function of $u$ or $\alpha$ in $L^2(\mathds{R}^3)$. Furthermore, for any $\Phi\in D(\psi(\bar{\gamma}))$ and $\Psi\in D(a(\bar{\gamma}))$, with $\gamma\in L^2(\mathds{R}^3)$, $C(u,\alpha)\Phi\in D(\psi(\bar{\gamma}))$, $C(u,\alpha)\Psi\in D(a(\bar{\gamma}))$ and the following identities hold:
      \begin{equation*}
        C(u,\alpha)^*\psi(\bar{\gamma})C(u,\alpha)\Phi=\psi(\bar{\gamma})\Phi+\braket{\gamma}{u}_2\Phi\; ;\;C(u,\alpha)^* a(\bar{\gamma})C(u,\alpha)\Psi=a(\bar{\gamma})\Psi+\braket{\gamma}{\alpha}_2\Psi
      \end{equation*}
      \begin{equation*}
        C(u,\alpha)^*\psi^*(\gamma)C(u,\alpha)\Phi=\psi^*(\gamma)\Phi+\braket{u}{\gamma}_2\Phi\; ;\; C(u,\alpha)^* a^*(\gamma)C(u,\alpha)\Psi=a^*(\gamma)\Psi+\braket{\alpha}{\gamma}_2\Psi
      \end{equation*}
\item\label{qt.item:6} Let $u,\alpha:t\to u(t),\alpha(t)\in \mathscr{C}^1(\mathds{R},L^2)$. Then $C(u(t),\alpha(t))$ is strongly differentiable in $t$ from $D(N_1^{1/2})\cap D(N_2^{1/2})$ to $\mathscr{H}$. The derivative is given by
  \begin{equation*}
    \begin{split}
      \dtot{}{t}C(u(t),\alpha(t))=\bigl[ \psi^*(\dot{u})-\psi(\dot{\bar{u}})-i\Im\braket{u}{\dot{u}}_2+a^*(\dot{\alpha})-a(\dot{\bar{\alpha}})-i\Im\braket{\alpha}{\dot{\alpha}}_2  \bigr]C(u(t),\alpha(t))
    \end{split}
  \end{equation*}
where $\dot{u}$, $\dot{\alpha}$ are the time derivatives respectively of $u$ and $\alpha$.
\item Let $u$, $\alpha\in L^2(\mathds{R}^3)$. Then for all $\delta\in\mathds{R}$, we have the following invariances: $C(u,\alpha)\Phi\in D(N_2^\delta)$ $\forall \Phi\in D(N_2^\delta)$; $C(u,\alpha)\Phi\in D(N_1^\delta)$ $\forall \Phi\in D(N_1^\delta)$; $C(u,\alpha)\Phi\in D(N^\delta)$ $\forall \Phi\in D(N^\delta)$.
\item We recall the definition of $U_{01}(t)\equiv \exp(i\Delta t/2)$ and $U_{02}(t)\equiv \exp(-i\omega t)$ given in the previous chapter. They are unitary operators on $L^2(\mathds{R}^3)$. Now define $\tilde{u}(t)=U_{01}^*(t)u(t)$, $\tilde{\alpha}(t)=U_{02}^*(t)\alpha(t)$ for all $u,\alpha\in \mathscr{C}^0(\mathds{R},L^2(\mathds{R}^3))$. Then the following equality holds $\forall\Phi\in\mathscr{H}$ and $t\in\mathds{R}$: $    U_0^*(t)C(u(t),\alpha(t))U_0(t)=C(\tilde{u}(t),\tilde{\alpha}(t))$.
  \end{enumerate}
\end{proposition}
\subsection{The quantum fluctuations.}
\label{sec:quant-fluct}
  We define the operator $V(t)\equiv V_{--}(t)+V_{-+}(t)+V_{+-}(t)+V_{++}(t)+V_0(t)$ on $D(V(t))$, where ($-$ is related to annihilation, and $+$ to creation):
\begin{gather*}
  V_{\#\#}(t)=\int\de{x}\ide{k}v_{\#\#}(t,x,k)\psi^{\#}(x)a^{\#}(k)\; ,\\
  V_{0}(t)=\int\ide{x}((2\pi)^{-3/2}\mathcal{F}^{-1}(\chi)*A(t))\psi^*\psi=\de\Gamma_p\bigl(((2\pi)^{-3/2}\mathcal{F}^{-1}(\chi)*A(t))\bigr)\otimes 1\; ,
\end{gather*}
$v_{\#\#}\in \mathscr{C}^0(\mathds{R},L^2(\R^3\otimes\R^3))$ and $A\in \mathscr{C}^0(\mathds{R},L^3(\R^3))$. Let $u\in \mathscr{C}^0(\mathds{R},L^2(\mathds{R}^3))$, then we can write explicitly $v_{\#-}=f_0(k)e^{ik\cdot x}u^{\#}(t,x)$; $v_{\#+}=f_0(k)e^{-ik\cdot x}u^{\#}(t,x)$. The following lemma can be easily proved applying the commutator theorem~\citep[Section~X.5]{MR0493420}.
\begin{lemma}\label{sec:evol-quant-fluct-1}
   For all $t\in\mathds{R}$, $V(t)$ is essentially self-adjoint on any core of $N$.
\end{lemma}
We would like to define the evolution operator of the quantum fluctuations as the evolution group generated by $H_2=H_0+V(t)$; however this could be done with mathematical rigour only passing to the so-called interaction representation. Then we define $\widetilde{V}(t)=U^*_0(t)V(t)U_0(t)$. Observe that $\V(t)$ is essentially self-adjoint on $D(N)$ for all $t\in\mathds{R}$, since $U_0(t)$ maps $D(N)$ into itself. Explicitly we can write
\begin{equation*}
 \V(t)=\V_{--}(t)+\V_{-+}(t)+\V_{+-}(t)+\V_{++}(t)+\V_0(t) 
\end{equation*}
with
\begin{equation*}
 \V_{\#\#}(t)=\int\de{x}\ide{k}\v_{\#\#}(t,x,k)\psi^\#(x)a^\#(k)\; ,
\end{equation*}
\begin{equation*}
  \V_0=\de\Gamma_p\Bigl(U_{01}^*(t)((2\pi)^{-3/2}\mathcal{F}^{-1}(\chi)*A(t))(\cdot)U_{01}(t)\Bigr)\otimes 1
\end{equation*}
and
\begin{equation*}
 \v_{--}(x,k)=U_{01}(t)U_{02}(t)v_{--}(x,k)=\bar{\v}_{++}(x,k)\; , 
\end{equation*}
\begin{equation*}
  \v_{-+}(x,k)=U_{01}(t)U_{02}^*(t)v_{-+}(x,k)=\bar{\v}_{+-}(x,k)\; .
\end{equation*}
$\de\Gamma_p(X)$ is the second quantization on $\mathscr{F}_s(p)$ of the operator $X$ on $L^2(\mathds{R}^3)$. By means of standard estimates the following lemma can be proved:
\begin{lemma}\label{sec:evol-quant-fluct-3}
    $\forall \delta\in\R$, $\V(t)$ belongs to $\mathscr{B}(\delta+2;\delta)$; furthermore is norm continuous as a function of $t$. We have in fact the following estimates:
  \begin{align*}
     \norm{\V_{--}\Phi}_\delta^2\leq&\frac{1}{2}c_\delta\norm{v_{--}(t)}^2_2\braket{\Phi}{(N+1)^{\delta+2}\Phi}\\
        \norm{\V_{-+}\Phi}_\delta^2\leq& \norm{v_{-+}(t)}^2_2\Bigl(\frac{1}{2}\braket{\Phi}{(N+1)^{\delta+2}\Phi}+\braket{\Phi}{(N+1)^{\delta+1}\Phi}\Bigr)\\
\norm{\V_{+-}\Phi}_\delta^2\leq& \norm{v_{+-}(t)}^2_2\Bigl(\frac{1}{2}\braket{\Phi}{(N+1)^{\delta+2}\Phi}+\braket{\Phi}{(N+1)^{\delta+1}\Phi}\Bigr)\\
\norm{\V_{++}\Phi}_\delta^2\leq& c_{-\delta}\norm{v_{++}(t)}^2_2\Bigl(\frac{1}{2}\braket{\Phi}{(N+1)^{\delta+2}\Phi}+2\braket{\Phi}{(N+1)^{\delta+1}\Phi}+\braket{\Phi}{(N+1)^{\delta}\Phi}\Bigr)\\
  \norm{\V_0\Phi}_\delta^2\leq& (2\pi)^{-3}\norm{(\mathcal{F}^{-1}(\chi)*A(t))}_\infty^2\braket{\Phi}{(N+1)^{\delta+2}\Phi}\; ,
  \end{align*}
where $c_\delta=1$ if $\delta\geq 0$, $c_\delta=3^{\ass{\delta}}$ otherwise.
\end{lemma}
To construct the evolution operator $\U_2(t,s)$ generated by $\V(t)$, we will use the Dyson series. However in order to do that we have to introduce a cut off in the total number of particles: let $\sigma_1\in \mathscr{C}^1(\mathds{R}^+)$, positive and decreasing, $\sigma_1(s)=1$ if $s\leq 1$, $\sigma_1(s)=0$ if $s\geq 2$; define $\sigma_{\nu}$ the operator $\sigma_1(N/\nu)$ in $\mathscr{H}$. Then we set 
\begin{equation*}
  \V_{\nu}(t)=\sigma_{\nu}\V(t)\sigma_{\nu}\; ,
\end{equation*}
for all $\nu\geq 1$.
\begin{lemma}
\label{sec:evol-quant-fluct-4}
Let $\V_{\nu}(t)$ be defined as above, then:
\begin{enumerate}[i.]
\item\label{qf.item:11}  $\V_{\nu}(t)$ satisfies lemma~\ref{sec:evol-quant-fluct-3}, with uniform bound in $\nu$. Furthermore $\V_{\nu}(t)$ is in $\mathscr{B}(\delta;\delta)$ for all $\delta\in\mathds{R}$ and is norm continuous as a function of $t$.
\item\label{qf.item:12} For all $\delta$ in $\R$, $\V_{\nu}(t)\to\V(t)$ when $\nu$ goes to infinity, in norm on $\mathscr{B}(\delta+2+\varepsilon;\delta)$, $\varepsilon>0$, and strongly in $\mathscr{B}(\delta+2;\delta)$, uniformly in $t$ on bounded intervals.
\end{enumerate}
\end{lemma}
\begin{proof}
  \begin{enumerate}[i.]
  \item Observe that $\sigma_{\nu}$ belongs to $\mathscr{B}(\delta;\delta')$ for all $\delta$ and $\delta'$ and $\norm{\sigma_{\nu}\Phi}_{\delta'}^2\leq c(\nu)\norm{\Phi}_{\delta}$, with $c(\nu)=\underset{p+n\leq 2\nu}{\mathrm{sup}}\bigl[\sigma_1^2\Bigl(\frac{p+n}{\nu}\Bigr)(p+n+1)^{\delta'-\delta}\bigr]$. Obviously if $\delta'\leq\delta$, $c(\nu)\leq 1$ for all $\nu\geq 1$, and we have a uniform bound in $\nu$. The result follows using lemma~\ref{sec:evol-quant-fluct-3}.
\item Strong convergence of $\V_{\nu}(t)$ to $\V(t)$ in $\mathscr{B}(\delta+2;\delta)$ follows from the obvious strong convergence on $\mathscr{C}_0(N_1,N_2)$, since $\V_{\nu}(t)$ is bounded in $\mathscr{B}(\delta+2;\delta)$ uniformly in $\nu$. Norm convergence on $\mathscr{B}(\delta+2+\varepsilon;\delta)$ follows from the fact that $(1-\sigma_{\nu})(N+1)^{-\varepsilon}$ goes to zero in norm as an operator in $\mathscr{H}$.
  \end{enumerate}
\end{proof}
The unitary group $\U_{2;\nu}(t,s)$ is defined by means of a Dyson series:
\begin{equation*}
 \U_{2;\nu}(t,s)=\sum_{m=0}^\infty(-i)^m\int_s^t\de t_1\int_s^{t_1}\ide{t_2}\dotsi\int_s^{t_{m-1}}\ide{t_m}\V_{\nu}(t_1)\dotsc \V_{\nu}(t_m)\; .
\end{equation*}
Using previous lemma we see that the series converge in norm on $\mathscr{B}(\delta;\delta)$ and $\U_{2;\nu}(t,s)$ is continuous and differentiable in norm with respect to $t$ on $\mathscr{B}(\delta;\delta)$ for all real $\delta$. We list below some useful properties of the family $\U_{2;\nu}(t,s)$, whose proof is immediate since $\V_\nu\in\mathscr{B}(\delta;\delta)$ for all $\delta\in\R$:
\begin{lemma}\begin{enumerate}[i.]
\item $\U_{2;\nu}(s,s)=1$, $\U_{2;\nu}(t,r)\U_{2;\nu}(r,s)=\U_{2;\nu}(t,s)$ for all $r,s,t\in\mathds{R}$.
\item $\U^*_{2;\nu}(t,s)=\U_{2:\nu}(s,t)$, and $\U_{2;\nu}(t,s)$ is unitary in $\mathscr{H}$.
\item $\U_{2;\nu}(t,s)$ is norm differentiable on $\mathscr{B}(\delta;\delta)$ for all real $\delta$, and
\begin{align*}i\dtot{}{t}\U_{2;\nu}(t,s)&=\V_\nu(t)\U_{2;\nu}(t,s)\; ;\; i\dtot{}{s}\U_{2;\nu}(t,s)=-\U_{2;\nu}(t,s)\V_\nu(s)\; .\end{align*}
  \end{enumerate}
\end{lemma}
The operators $\U_{2;\nu}(t,s)$ also satisfy the following crucial boundedness property:
\begin{lemma}\label{sec:evol-quant-fluct-5}
The operator $\U_{2;\nu}(t,s)$ is bounded on $\H^{\delta}$ uniformly in $\nu$ for all real $\delta$. More precisely:
 \begin{equation}\label{qf.eq:6}
   \norm{\U_{2;\nu}(t,s)}_{\mathscr{B}(\delta;\delta)}\leq\exp\biggl\{\frac{\ass{\delta}}{2}\biggl(\ln 3+ \sqrt{2}c_2(\delta)\ass{\int_s^t\ide{\tau}\norm{v_{--}(\tau)}_2}\biggr)\biggr\}\; ,
 \end{equation}
with $c_2(\delta)=\max (4,3^{\ass{\delta}/2}+1)$.
\end{lemma}
This lemma can be proved by the same argument used to prove the last point of proposition~\ref{sec:quantum-theory-3}, with an operator $h(N_1+N_2)$ on $\mathscr{H}$ instead of $h(N_2)$ on $\mathscr{H}_p$.

We are ready to define the fluctuations evolution operator $\U_2(t,s)$.
\begin{proposition}\label{sec:evol-quant-fluct-6}
  There is a family of operators $\U_2(t,s)$ satisfying the following properties:
  \begin{enumerate}[i.]
  \item\label{qf.item:13} for all $\delta\in\R$, $\U_2(t,s)$ is bounded and strongly continuous with respect to $t$ and $s$ on $\H^\delta$ and satisfies
    \begin{equation}
\label{qf.eq:8}
      \norm{\U_{2}(t,s)}_{\mathscr{B}(\delta;\delta)}\leq\exp\biggl\{\frac{\ass{\delta}}{2}\biggl(\ln 3+ \sqrt{2}c_2(\delta)\ass{\int_s^t\ide{\tau}\norm{v_{--}(\tau)}_2}\biggr)\biggr\}\; ,
    \end{equation}
with $c_2(\delta)=\max(4,3^{\ass{\delta}/2}+1)$.
\item\label{qf.item:14} $\U_2(t,s)$ is unitary in $\mathscr{H}$.
\item\label{qf.item:15} $\U_2(s,s)=1$, $\U_2(t,r)\U_2(r,s)=\U_2(t,s)$ for all $r$, $s$ and $t$ in $\mathds{R}$.
\item\label{qf.item:16} For all $\delta\in\R$, $\U_2(t,s)$ is strongly differentiable from $\H^{\delta+2}$ to $\H^\delta$; in particular is strongly differentiable from $D(N)$ to $\mathscr{H}$. Furthermore:
\begin{align*}i\dtot{}{t}\U_2(t,s)&=\V(t)\U_2(t,s)\; ;\; i\dtot{}{s}\U_2(t,s)=-\U_2(t,s)\V(s)\; .\end{align*}
\item\label{qf.item:17} For all $\Psi\in D(N)$ and $\Phi\in\mathscr{H}$, $i\partial_t\braket{\Psi}{\U_2(t,s)\Phi}=\braket{\V(t)\Psi}{\U_2(t,s)\Phi}$.
\item\label{qf.item:18} Let $U_2(t,s)=U_0(t)\U_2(t,s)U_0^{-1}(s)$; for all $\Psi\in D(N)\cap D(H_0)$, $\Phi\in\mathscr{H}$
  \begin{equation*}
     i\partial_t\Bigl\langle\Psi,U_2(t,s)\Phi\Bigr\rangle=\Bigl\langle\Bigl(H_0+V(t)\Bigr)\Psi,U_2(t,s)\Phi\Bigr\rangle\; .
  \end{equation*}
  \end{enumerate}
\end{proposition}
\begin{proof}
  \begin{enumerate}[i.]
  \item For all couples of positive integers $\nu$ and $\nu'$, write
    \begin{equation*}
      \U_{2;\nu}(t,s)-\U_{2;\nu'}(t,s)=-i\int_{s}^t\ide{\tau} \U_{2;\nu'}(t,\tau)(\V_\nu(\tau)-\V_{\nu'}(\tau))\U_{2;\nu}(\tau,s) ,
    \end{equation*}
as a Riemann integral in norm on $\mathscr{B}(\delta;\delta)$ for all $\delta$. Then, using first part of lemma~\ref{sec:evol-quant-fluct-4} and equation~\eqref{qf.eq:6} we obtain
\begin{equation*}\begin{split}
  \norm{\U_{2;\nu}(t,s)-\U_{2;\nu'}(t,s)}_{\mathscr{B}(\delta+2+\varepsilon;\delta)}\leq\ass{t-s}e^{\gamma\ass{\int_s^t\ide{\tau}\norm{v_{--}(\tau)}_2}}\sup_{\tau\in[s,t]}\norm{\V_\nu(\tau)-\V_{\nu'}(\tau)}_{\mathscr{B}(\delta+2+\varepsilon;\delta)}\; ,
\end{split}\end{equation*}
where $\gamma$ depends on $\delta$ and $\varepsilon$. Utilizing then second part of lemma~\ref{sec:evol-quant-fluct-4}, we see that for all $\delta\in\R$, $\U_{2;\nu}(t,s)$ converges in norm on $\mathscr{B}(\delta+2+\varepsilon;\delta)$ when $\nu\to\infty$ uniformly in $t$ and $s$ on every compact interval. The resulting limit $\U_2(t,s)$ is continuous in the norm of $\mathscr{B}(\delta+2+\varepsilon;\delta)$ with respect to $t$ and $s$. The norm convergence just proved and the estimate~\eqref{qf.eq:6}, uniform in $\nu$, imply the strong convergence of $\U_{2;\nu}(t,s)$ to $\U_2(t,s)$ on $\mathscr{B}(\delta;\delta)$ uniformly in $t$ and $s$ on every compact interval. Consequently $\U_2(t,s)$ satisfies the estimate~\eqref{qf.eq:8} and is strongly continuous in $t$ and $s$.
\item The result follows from the unitarity of $\U_{2;\nu}(t,s)$ on $\H$ and from the strong convergence of $\U_{2;\nu}(t,s)$ and its adjoint $\U_{2;\nu}(s,t)$.
\item The result is an immediate consequence of the corresponding properties of $\U_{2;\nu}(t,s)$.
\item Write $\U_{2;\nu}(t,s)\Phi$, with $\Phi\in\H^{\delta+2}$, as a strong Riemann integral on $\H^\delta$:
  \begin{equation*}
       \U_{2;\nu}(t,s)\Phi=\Phi-i\int_{s}^t\ide{\tau}\V_\nu(\tau)\U_{2;\nu}(\tau,s)\Phi\; .
  \end{equation*}
Using point~\ref{qf.item:12}. of lemma~\ref{sec:evol-quant-fluct-4} and the strong convergence proved above we can go to the limit $\nu\to\infty$ in previous equation. The result then following from lemma~\ref{sec:evol-quant-fluct-3} and from point \ref{qf.item:13}. of this lemma.
\item Consider both $\Psi$ and $\Theta$ in $D(N)$, then using previous point:
\begin{equation}\label{qf.eq:10}
   \braket{\Psi}{\U_2(t,s)\Theta}-\braket{\Psi}{\Theta}=-i\int_s^t\ide{t'}\braket{\V(t')\Psi}{\U_2(t',s)\Theta}\; .
\end{equation}
Consider now $\{\Phi_j\}\in D(N)$ such that $\H-\lim_j\Phi_j=\Phi\in\H$. For all $\Phi_j$ equation~\eqref{qf.eq:10} holds, furthermore both $\V(t)\Psi$ and $\U_2(t,s)\Phi_j$ are uniformly bounded in $t$, so we use the dominated convergence theorem to go to the limit $j\to\infty$.
\item With the aid of previous point, we calculate explicitly, for $\Psi\in D(N)\cap D(H_0)$, $\Phi\in\H$ the derivative:
  \begin{equation*}\begin{split}
      i\partial_t\Bigl\langle\Psi,U_2(t,s)\Phi\Bigr\rangle=\Bigl\langle H_0U_0^{-1}(t)\Psi,\U_2(t,s)U_0^{-1}(s)\Phi\Bigr\rangle+\Bigl\langle\V(t)U_0^{-1}(t)\Psi,\U_2(t,s)U_0^{-1}(s)\Phi\Bigr\rangle\; ,
  \end{split}\end{equation*}
where the second term of the right hand side of the equality makes sense because $D(N)\cap D(H_0)$ is invariant under the action of $U_0^{-1}(t)$ since $N$ and $H_0$ commute. The result follows immediately.
  \end{enumerate}
\end{proof}
We want to emphasize that, even if $U_2(t,s)$ defined above is formally generated by $H_0+V(t)$, \emph{i.e.} formally satisfies the equation $i\dtot{}{t}U_2(t,s)=\bigl(H_0+V(t)\bigr)U_2(t,s)$, we can only assert that it is weakly differentiable in the sense make explicit in point~\ref{qf.item:18}. of the previous proposition. We are not able to formulate any strong differentiability property for $U_2$, and we need to use the interaction representation in order to take strong derivatives. However the following uniqueness result regarding $U_2$ can be proved:
\begin{lemma}
  Let $s\in\R$, $\Phi(\cdot)\in\mathscr{C}_{W}(\R,\H)$ with $\Phi(s)\equiv \Phi$, such that $i\partial_t\bigl\lvert\bigl\langle\Psi,\Phi(t)\bigr\rangle\bigr\rvert=\bigl\lvert\bigl\langle\bigl(H_0+V(t)\bigr)\Psi,\Phi(t)\bigr\rangle\bigr\rvert$ for all $\Psi\in D(N)\cap D(H_0)$ and $\Phi\in\H$. Then $\Phi(t)=U_2(t,s)\Phi$.
\end{lemma}

\subsection{The evolution $\W(t,s)$.}
\label{sec:evolution-wt-s}

We recall the definition of the unitary evolution between coherent states:
\begin{equation*}
  W(t,s)=C^*(u_\lambda(t),\alpha_\lambda(t))U(t-s)C(u_\lambda(s),\alpha_\lambda(s))e^{i\Lambda(t,s)}\; ,
\end{equation*}
where $\Lambda(t,s)$ is a phase function, and $(u(\cdot),\alpha(\cdot))$ is the $\mathscr{C}^0(\mathds{R},L^2(\mathds{R}^3)\otimes L^2(\mathds{R}^3))$ unique solution of~\eqref{eq:2} corresponding to initial data $(u(s),\alpha(s))\in L^2(\mathds{R}^3)\otimes L^2(\mathds{R}^3)$ (the existence of such solution has been established in section~\ref{sec:stat-main-results}).

In the interaction picture, we will write $\W(t,s)=U^*_0(t)W(t,s)U_0(s)$, so using the last point of proposition~\ref{sec:quantum-theory-8}:
\begin{equation*}
    \W(t,s)=C^*(\widetilde{u}_\lambda(t),\widetilde{\alpha}_\lambda(t))U_0^*(t)U(t-s)U_0(s)C(\widetilde{u}_\lambda(s),\widetilde{\alpha}_\lambda(s))e^{i\Lambda(t,s)}\; .
\end{equation*}
By lemma~\ref{cl.rem:1} $(\widetilde{u}(\cdot),\widetilde{\alpha}(\cdot))\in \mathscr{C}^1(\mathds{R},L^2(\mathds{R}^3)\otimes L^2(\mathds{R}^3))$, and
\begin{align*}
  i\partial_t \widetilde{u}(t)&=(2\pi)^{-3/2}U_{01}(-t)\Bigl(\mathcal{F}^{-1}(\chi) * A(t)\Bigr)u(t)\; \\
  i\partial_t \widetilde{\alpha}(t)&=\frac{(2\pi)^{-3/2}}{\sqrt{2}}U_{02}(-t)\Bigl(\omega^{-1/2}\chi \mathcal{F}(\ass{u}^2) (t)\Bigr)\; .
\end{align*}
By definition $\W(t,s)$ is unitary on $\mathscr{H}$ and such that $\W^*(t,s)=\W(s,t)$. Define now
\begin{equation*}
  Z(t)=C^*(\widetilde{u}_\lambda(t),\widetilde{\alpha}_\lambda(t))U_0^*(t)U(t)e^{i\Lambda(t,0)}\; \Rightarrow\; \W(t,s)=Z(t)Z^*(s)\; .
\end{equation*}
Define also the domains:
\begin{equation}
  \label{eq:19}
  \begin{aligned}
    \mathscr{D}&=\{\Psi\in D(N) | C(\widetilde{u}_\lambda(s),\widetilde{\alpha}_\lambda(s))\Psi\in D(H_0)\}\\
\mathscr{D}^\delta&=\{\Psi\in \mathscr{H}^\delta | C(\widetilde{u}_\lambda(s),\widetilde{\alpha}_\lambda(s))\Psi\in D(H_0)\}\; .
  \end{aligned}
\end{equation}
\begin{lemma}
  \label{sec:limit-lambdato-0-5}
$Z(t)$ is strongly differentiable from $D(H_0)\cap D(N_1^2+N_2)$ to $\mathscr{H}$. If $\Lambda(t,s)$ satisfies equation~\eqref{eq:12}, then for all $\Psi\in D(H_0)\cap D(N_1^2+N_2)$ we have $  i\partial_t Z(t)\Psi=\Bigl( U_0^*(t)H_{I}U_0(t)+\widetilde{V}(t) \Bigr)Z(t)\Psi$.
\end{lemma}
\begin{proof}
Let $\Psi\in D(H_0)\cap D(N_1^2+N_2)$, then $U(t)$ is differentiable on $\Psi$ since $D(H)\supseteq D(H_0)\cap D(N_1^2+N_2)$ by proposition~\ref{sec:quantum-theory-2}; furthermore by proposition~\ref{sec:quantum-theory-3} $U(t)\Psi\in D(H)\cap D(N_1^2+N_2)\subseteq D(H_0)$, so also $U_0^*(t)$ is differentiable on $U(t)\Phi$; finally by proposition~\ref{sec:quantum-theory-3} $U_0^*(t)U(t)\Psi\in D(H_0)\cap D(N_1^2+N_2)\subseteq D(N_1^{1/2})\cap D(N_2^{1/2})$, so we can differentiate each factor of $Z(t)$. Then for all $\Psi\in D(H_0)\cap D(N_1^2+N_2)$:
\begin{equation*}
  \begin{split}
      i\dtot{}{t}Z(t)\Psi=C^*(\widetilde{u}_\lambda(t),\widetilde{\alpha}_\lambda(t))U_0^*(t)\Bigl\{ U_0(t)\Bigl( -\psi^*(i\dot{\widetilde{u}}_\lambda)+\psi(i\dot{\bar{\widetilde{u}}}_\lambda)-\Im\braket{\widetilde{u}_\lambda}{\dot{\widetilde{u}}_\lambda}-a^*(i\dot{\widetilde{\alpha}}_\lambda)\\
+a(i\dot{\bar{\widetilde{\alpha}}}_\lambda)-\Im\braket{\widetilde{\alpha}_\lambda}{\dot{\widetilde{\alpha}}_\lambda}  \Bigr)U_0^*(t)+H_{I}-\dtot{}{t}\Lambda(t,0)\Bigr\}U(t)e^{i\Lambda(t,0)}\Psi\; .
  \end{split}
\end{equation*}
The result is then obtained by algebraic manipulation and using the fact that for all $f\in L^\infty(\mathds{R}^3,L^2(\mathds{R}^3))$, $u$, $\alpha\in L^2(\mathds{R}^3)$ and $\Phi\in D(N_1^2+N_2)$:
\begin{equation*}
  \begin{split}
    C^*(u,\alpha)\int\de x\ide{k}\Bigl(\bar{f}(x,k)a(k)+f(x,k)a^*(k)\Bigr)\psi^*(x)\psi(x) C(u,\alpha)\Phi= \int\de x \ide{k}\Bigl(\bar{f}(x,k)\\(a(k)+\alpha(k))+f(x,k)(a^*(k)+\bar{\alpha}(k))\Bigr)(\psi^*(x)+\bar{u}(x))(\psi(x)+u(x))\Phi\; .
  \end{split}
\end{equation*}

\end{proof}
\begin{lemma}
  \label{sec:limit-lambdato-0-6}
$Z^*(t)$ is strongly differentiable from $D(N_1^2+N_2)$ to $\mathscr{H}$, and if $\Lambda(t,s)$ satisfies equation~\eqref{eq:12}, then for all $\Phi\in D(N_1^2+N_2)$ we have $i\partial_t Z^*(t)\Phi=-Z^*(t)\Bigl( U_0^*(t)H_{I}U_0(t)+\widetilde{V}(t) \Bigr)\Phi$.
\end{lemma}
\begin{proof}
  Let $B\equiv U_0^*(t)H_{I}U_0(t)+\widetilde{V}(t)$. 
Then we can write for all $\Phi\in D(N_1^2+N_2)$:
\begin{equation*}
  i\partial_t \braket{Z^*(t)\Phi}{\Psi}=\braket{Z^*(t)B\Phi}{\Psi}\; ,
\end{equation*}
so integrating on both sides we find
\begin{equation*}
  i\Bigl( \braket{Z^*(t)\Phi}{\Psi} - \braket{Z^*(0)\Phi}{\Psi} \Bigr)=\int_0^t \ide{\tau}\braket{Z^*(\tau)B\Phi}{\Psi}\; ,
\end{equation*}
but since $Z^*(\tau)B\Phi$ is continuous in $\tau$ for all $\Phi\in D(N_1^2+N_2)$
\begin{equation*}
  i\dtot{}{t}Z^*(t)\Phi=-Z^*(t)B\Phi\; .
\end{equation*}

\end{proof}
Lemmas~\ref{sec:limit-lambdato-0-5} and~\ref{sec:limit-lambdato-0-6} prove the following proposition:
\begin{proposition}
  $\widetilde{W}(t,s)$ is strongly differentiable in $t$ from $\mathscr{D}$ to $\H$; $\W^*(t,s)$ is strongly differentiable in $t$ from $D(N_1^2+N_2)$ to $\H$. More precisely if
  \begin{equation*}
    \Lambda(t,s)=-\frac{1}{2}(2\pi)^{-3/2}\lambda^{-2}\int_s^t\de{t'}\int\ide{x}(\mathcal{F}^{-1}(\chi)*A(t'))\bar{u}(t')u(t')\; ,
  \end{equation*}
 then for all $\Psi\in \mathscr{D}$, $\Phi\in D(N_1^2+N_2)$
 \begin{equation}
   \label{conv.eq:14}  
\begin{aligned}
    i\dtot{}{t}&\widetilde{W}(t,s)\Psi=\Bigl( U_0^*(t)H_{I}U_0(t)+\widetilde{V}(t) \Bigr) \widetilde{W}(t,s)\Psi\\
    i\dtot{}{t}&\widetilde{W}^*(t,s)\Phi=-\widetilde{W}^*(t,s) \Bigl( U_0^*(t)H_{I}U_0(t)+\widetilde{V}(t) \Bigr)\Phi\; .
  \end{aligned}
 \end{equation}
\end{proposition}
\subsection{Proof of Theorem~\ref{sec:conv-quant-evol-1}.}
\label{sec:convergence-wt-s}
   We prove the existence of the limit on $\mathscr{D}^{\delta}$ with $\delta\geq 4$, dense in $\H$ (see equation~\eqref{eq:19}). $\widetilde{W}$ is strongly differentiable on such domain and $\widetilde{W}[\mathscr{D}^{\delta}]\subseteq \mathscr{H}^{\delta/2}$; while $\widetilde{U}_2$ is strongly differentiable on $\mathscr{H}^{\delta/2}$, when $\delta\geq 4$. Then we can write the following inequalities for all $\Phi\in \mathscr{D}^{\delta}$, the integrals making sense as strong Riemann integrals on $\H$:
   \begin{equation*}\begin{split}
     \Bigl\lVert\Bigl(\W(t,s)-\U_2(t,s)\Bigr)\Phi\Bigr\rVert^2=-2\,\Re\, \Bigl\langle \Phi,\int_s^t\ide{\tau}\dtot{}{\tau}\U_2^*(\tau,s)\W(\tau,s)\Phi\Bigr\rangle\\
=2\,\Im\,\int_s^t\ide{\tau}\Bigl\langle H_{I}U_0(\tau)\U_2(\tau,s)\Phi,U_0(\tau)\W(\tau,s)\Phi\Bigr\rangle\leq 2\lambda \norm{f_0}_2\norm{\Phi}\ass{\int_s^t\ide{\tau}\norm{\U_2(\tau,s)\Phi}_{\mathscr{H}^4}}\\
\leq 2\lambda\norm{f_0}_2\ass{\int_s^t\ide{\tau}\exp\biggl\{2\biggl(\ln 3+ 10\sqrt{2}\ass{\int_s^\tau\ide{\tau'}\norm{v_{--}(\tau')}_2}\biggr)\biggr\}}\norm{\Phi}\norm{\Phi}_{\delta}
\end{split}\end{equation*}
that tends to zero when $\lambda\to 0$, uniformly in $t$ and $s$ on compact intervals. In the first inequality we have used the corollary of lemma~\ref{sec:quantum-theory-6}, in the second inequality proposition~\ref{sec:evol-quant-fluct-6}.

\section{Classical limit of annihilation and creation operators.}
\label{sec:class-limit-annih-3}

In this section we develop the proofs of theorem~\ref{sec:class-limit-annih-1} (section~\ref{sec:proof-theorem}) and~\ref{sec:class-limit-annih-2} (section~\ref{sec:proof-th}). In order to do that we have to find a bound for $\norm{\W(t,s)\Phi}_\delta$ that is finite when $\lambda\to 0$, this is done in proposition~\ref{sec:useful-estimates-3}; in proposition~\ref{sec:useful-estimates-4} we prove a result on the $\U_2(t,s)$-evolution of quantum fields that enables us to improve the rate of convergence of averages of creation and annihilation operators to order $\lambda^2$.

As discussed above, a bound of $\norm{\W(t,s)\Phi}_\delta$ is needed. Such bound has to converge when $\lambda\to 0$. A $\lambda$-divergent bound is quite easy to prove, using the following preliminary result on $\mathscr{F}_s(L^2(\mathds{R}^3))$ (only for this lemma $a^\#(g)$ are the annihilation and creation operators, $N$ the number operator and $C(g)$ the Weyl operator of $\mathscr{F}_s(L^2(\mathds{R}^3))$):
\begin{lemma}[on $\mathscr{F}_s(L^2(\mathds{R}^3))$]\label{sec:useful-estimates-10}
  Let $b\geq 1/2$. Then for all $m=1,2,\dotsc$ and $\Psi\in D(N^m)$ we have
  \begin{equation*}
    \begin{split}
      \norm{(N+b)^mC(g)\Psi}\leq 6^{m/2}\norm{\prod_{j=0}^{m-1}(N+b+\norm{g}_2^2+j)\Psi}\leq 6^{m/2}(1+2(m-1))^m(1+2\norm{g}_2^2)^m\\
\norm{(N+b)^m\Psi}\; .
    \end{split}
  \end{equation*}
\end{lemma}
\begin{proof}
  Using properties of Weyl operators we can write $C^*(g) (N+b) C(g)= N+b+\norm{g}_2^2+a(\bar{g})+a^*(g)$, and such equality holds on $D(N)$. So if $m=1$, $\Psi\in D(N^m)$:
\begin{equation*}
  \begin{split}
    \norm{(N+b)C(g)\Psi}^2\leq 3\braket{\Psi}{\Bigl((N+b+\norm{g}_2^2)^2+2\norm{g}_2^2N+\norm{g}_2^2\Bigr)\Psi}\; .
  \end{split}
\end{equation*}
Now if $b\geq 1/2$ we have $2\norm{g}_2^2N+\norm{g}_2^2\leq (N+b+\norm{g}_2^2)^2$. Suppose the result is verified for $m$, and verify it for $m+1$. Let $  h_m(N)=\prod_{j=0}^{m-1}(N+b+\norm{g}_2^2+j)$. Then
\begin{equation*}
  \begin{split}
    \norm{(N+b)^{m+1}C(g)\Psi}^2\leq 3\braket{\Psi}{h_m(N+1)^2\Bigl((N+b+\norm{g}_2^2)^2+2N\norm{g}_2^2+\norm{g}_2^2\Bigr)\Psi}\; .
  \end{split}
\end{equation*}
\end{proof}
\begin{lemma}\label{sec:useful-estimates-7}
  $C(u,\alpha)$ maps $\mathscr{H}^{2\delta}$ into itself for any positive $\delta$. In particular, let $u,\alpha\in L^2$, $\delta\geq 0$, $\Phi\in\mathscr{H}^{2\delta}$; then
  \begin{equation}\label{west.eq:12}
    \norm{C(u,\alpha)\Phi}_{2\delta}\leq c_5(\delta) \bigl(1+2\norm{u}_2^2+2\norm{\alpha}_2^2\bigr)^{\delta} \norm{\Phi}_{2\delta}\; ,
  \end{equation}
with $c_5(\delta)$ a positive constant depending on $\delta$.
\end{lemma}
\begin{proof}
  The result is a direct consequence of lemma~\ref{sec:useful-estimates-10} when $\delta$ is an integer. By interpolation we extend it to all real $\delta$, with
  \begin{equation*}
    c_5(\delta)=6^{\delta/2}(1+2(d_- -1))^{\frac{d_-(d_+ -\delta)}{d_+-d_-}}(1+2(d_+ -1))^{\frac{d_+(\delta-\delta_-)}{d_+-d_-}}\; ,
  \end{equation*}
where $d_-=\max_{m\in\mathds{N}}\{m\leq\delta\}$, $d_+=\min_{m\in\mathds{N}}\{m\geq\delta\}$.
\end{proof}
Using lemma~\ref{sec:useful-estimates-7} and proposition~\ref{sec:quantum-theory-3} the following bound is proved.
\begin{proposition}
  \label{sec:useful-estimates-2}
Let $\Phi\in\mathscr{H}^{4\delta}$, with positive integer $\delta$ and $\lambda\leq 1$; then
\begin{equation*}
  \norm{\W(t,s)\Phi}_{2\delta}\leq L_\delta(t,s)\lambda^{-6\delta}\exp\Bigl\{\ass{\delta} c_1(\delta) \lambda\norm{f_0}_2 \ass{t-s}\Bigr\}\norm{\Phi}_{4\delta}
\end{equation*}
with $c_1(\delta)=\max(3,1+2^{\ass{\delta}})$ and
\begin{equation*}
  L_\delta(t,s)=c_5(2\delta)c_5(\delta) \bigl(1+2\norm{u(t)}_2^2+2\norm{\alpha(t)}_2^2\bigr)^{\delta}\bigl(1+2\norm{u(s)}_2^2+2\norm{\alpha(s)}_2^2\bigr)^{2\delta}\; .
\end{equation*}
\end{proposition}
To obtain a $\lambda$-convergent bound we have to restrict to a narrower space than $\mathscr{H}^{4\delta}$. We also have to introduce a regularized evolution. From now on we will use the notation $\widetilde{H}_I(t)=U_0^*(t)H_IU_0(t)$. We also define the orthogonal projectors $(\mathbb{N}_1)_{\leq\nu}$ and $(\mathbb{N}_2)_{\leq\nu}$ as following:
\begin{gather*}
  ((\mathbb{N}_1)_{\leq\nu}\Phi)_{p,n}=\left\{
    \begin{aligned}
      \Phi_{p,n}&\text{ if $p\leq \nu$}\\
      0&\text{ if $p>\nu$}
    \end{aligned}
\right .\; ;\; 
  ((\mathbb{N}_2)_{\leq\nu}\Phi)_{p,n}=\left\{
    \begin{aligned}
      \Phi_{p,n}&\text{ if $n\leq \nu$}\\
      0&\text{ if $n>\nu$}
    \end{aligned}
\right .\; .
\end{gather*}
We define $\mathbb{R}_\nu=(\mathbb{N}_1)_{\leq\nu}(\mathbb{N}_2)_{\leq\nu}$, so $\mathbb{R}_\nu\Bigl(\widetilde{H}_I(t)+\V(t)\Bigr)\mathbb{R}_\nu$ is bounded in $\mathscr{H}$. Then by means of a Dyson series we obtain:
\begin{equation*}
  \begin{split}
      \W_\nu(t,s)=\sum_{m=0}^\infty(-i)^m\int_s^t\de t_1\int_s^{t_1}\de t_2\dotsm \int_s^{t_{m-1}}\de t_m \mathbb{R}_\nu\Bigl(\widetilde{H}_I(t_1)+\V(t_1)\Bigr)\mathbb{R}_\nu \dotsm \mathbb{R}_\nu\Bigl(\widetilde{H}_I(t_m)\\+\V(t_m)\Bigr)\mathbb{R}_\nu\; .
  \end{split}
\end{equation*}
\begin{lemma}\label{sec:useful-estimates-1}
\begin{enumerate}[i.]
\item $\W_{\nu}(s,s)=1$, $\W_{\nu}(t,r)\W_{\nu}(r,s)=\W_{\nu}(t,s)$ for all $r,s,t\in\mathds{R}$.
\item $\W^*_\nu(t,s)=\W_\nu(s,t)$, and $\W_{\nu}(t,s)$ are unitary in $\mathscr{H}$.
\item $\W_{\nu}(t,s)$ is strongly differentiable on $\mathscr{H}$ and
\begin{align*}i\dtot{}{t}&\W_{\nu}(t,s)=\mathbb{R}_\nu\Bigl(\widetilde{H}_I(t)+\V(t)\Bigr)\mathbb{R}_\nu\W_{\nu}(t,s)\\
i\dtot{}{s}&\W_{\nu}(t,s)=-\W_{\nu}(t,s) \mathbb{R}_\nu\Bigl(\widetilde{H}_I(s)+\V(s)\Bigr)\mathbb{R}_\nu\; .\end{align*}
\item Let $\Phi\in\mathscr{H}^{2\delta}$, $\delta\in\mathds{R}$, $\nu\geq 1$. Then
  \begin{equation*}\begin{split}
  \norm{\W_\nu(t,s)\Phi}_{2\delta}\leq \exp\biggl\{\sqrt{\nu}\lambda\ass{\delta}{c}_3(\delta)\norm{f_0}_2\ass{t-s}+\ass{\delta}\biggl(\ln 3+{c}_4(\delta)\ass{\int_s^t\ide{\tau}\norm{v_{--}(\tau)}_2}\biggr)\biggr\}\\\norm{\Phi}_{2\delta}\; ,
 \end{split}\end{equation*}
with ${c}_3(\delta)=\max(5/2,2^{\ass{\delta}}+1/2)$, ${c}_4(\delta)=\max(4,3^{\ass{\delta}}+1)$.
\end{enumerate}
\end{lemma}
\begin{proof}
Only the last point is not a direct consequence of the definition. For all $\Phi\in \mathscr{H}$,
\begin{equation*}
  M(t,s)=\frac{1}{2}\norm{(N+3)^{-\delta}\W_\nu(t,s)\Phi}^2\; ,
\end{equation*} 
with $\delta\geq 1$, differentiable in $t$ and $s$. Set $(N+3)^{-\delta}\equiv h(N)$, then
\begin{align*}
  \dtot{}{t}M(t,s)=&\Im\braket{h(N)\W_\nu(t,s)\Phi}{\mathbb{R}_\nu \widetilde{H}^-_I(t) \mathbb{R}_\nu\Bigl(h(N-1)h(N)^{-1}-1\Bigr)h(N)\W_\nu(t,s)\Phi}\\
+&\Im\braket{\mathbb{R}_\nu \widetilde{H}^-_I(t) \mathbb{R}_\nu\Bigl(h(N)h(N-1)^{-1}-1\Bigr)h(N) \W_\nu(t,s)\Phi}{h(N)\W_\nu(t,s)\Phi}\\
+&\Im\braket{h(N) \W_\nu(t,s)\Phi}{\V_{--;\nu}(t)\Bigl(h(N-2)h(N)^{-1}-1\Bigr)h(N) \W_\nu(t,s)\Phi}\\
+&\Im\braket{\V_{--;\nu}(t)\Bigl(h(N)h(N-2)^{-1}-1\Bigr)h(N) \W_\nu(t,s) \Phi}{h(N) \W_\nu(t,s)\Phi}
\end{align*}
The last two terms of the right hand side are bounded by lemma~\ref{sec:evol-quant-fluct-5}.
\begin{equation*}
  \begin{split}
    \ass{\dtot{}{t}M(t,s)}\leq 2\lambda\norm{f_0}_2\Bigl[\nnorm{\mathbb{R}_\nu N_1\sqrt{N_2}\biggl(h(N)h(N-1)^{-1}-1\biggr)}\\
+\nnorm{\mathbb{R}_\nu N_1\sqrt{N_2}\biggl(h(N-1)h(N)^{-1}-1\biggr)}\Bigr]M(t,s)+\sqrt{2}\norm{v_{--}}_2\delta(3^\delta+1)M(t,s)\; .
  \end{split}
\end{equation*}
Furthermore
\begin{equation*}
  \begin{split}
   K\equiv \Biggl[\dotsc\Biggr]\leq \delta\nu\sqrt{\nu}\biggl(\frac{(2\nu+3)^\delta}{(2\nu+2)^\delta}+\frac{(2\nu+3)^{\delta-1}}{(2\nu+3)^\delta}\biggr)\leq \delta\sqrt{\nu}\Bigl( 2^{\delta} +\frac{1}{2}\Bigr)\; .
  \end{split}
\end{equation*}
Applying now Gronwall's Lemma
\begin{equation*}\begin{split}
  \norm{(N+3)^{-\delta}\W_\nu(t,s)\Phi}\leq \exp\biggl\{\sqrt{\nu}\delta\bigl( 2^{\delta}+\frac{1}{2}\bigr)\lambda\norm{f_0}_2(t-s)+\frac{1}{\sqrt{2}}\delta(3^\delta+1)\\\ass{\int_s^t\ide{\tau}\norm{v_{--}(\tau)}_2}\biggr\}\norm{(N+3)^{-\delta}\Phi}^2\; ,
\end{split}\end{equation*}
for all $\delta\geq 1$. Interpolating between $\delta=0$ and $\delta=1$ we obtain the result for all $\delta\geq 0$; by duality we extend the result to all $\delta\in\mathds{R}$.
\end{proof}
\begin{proposition}
  \label{sec:useful-estimates-3}
For all positive $\delta$ exists a $\delta'>\delta$ such that $\W(t,s)$ maps $\mathscr{H}^{\delta'}$ into $\mathscr{H}^{\delta}$. In particular let $\Phi\in\mathscr{H}^{\delta'}$. Then for all $\lambda\leq 1$, $\delta'= \max(4,6\delta+3)$:
\begin{equation*}
  \norm{\W(t,s)\Phi}^2_\delta\leq \Bigl(\lambda\norm{f_0}_2\ass{t-s}+\mathcal{L}_1(t,s)\Bigr)e^{\lambda\mathcal{C}_1\ass{t-s}+\mathcal{L}_2(t,s)}\norm{\Phi}^2_{\delta'}\; ,
\end{equation*}
where $\mathcal{C}_1$ is a positive constant depending on $\delta$; $\mathcal{L}_j(t,s)$, $j=1,2$, positive functions depending also on $\delta$.
\end{proposition}
\begin{proof}
  Let $\Phi\in \mathscr{H}^{\delta'}$, with $\delta'\geq 4$. Due to the properties of $\W(t,s)$ and $\W_\nu(t,s)$ all the steps of the following proof are well defined, and the integrals make sense as strong Riemann integrals on $\mathscr{H}$. We evaluate separately each term of the right hand side of the identity
  \begin{equation*}
    \begin{split}
      \braket{\W(t,s)\Phi}{(N+1)^\delta\W(t,s)\Phi}=\braket{\W_\nu(t,s)\Phi}{(N+1)^\delta\W_\nu(t,s)\Phi}+\braket{\W(t,s)\Phi}{(N+1)^\delta\\\Bigl(\W(t,s) -\W_\nu(t,s)\Bigr)\Phi}+\braket{\Bigl(\W(t,s) -\W_\nu(t,s)\Bigr)\Phi}{(N+1)^\delta \W_\nu(t,s)\Phi}\; .
    \end{split}
  \end{equation*}
The estimate for the first one is provided by lemma~\ref{sec:useful-estimates-1}. Consider now the second term:
\begin{equation*}
  \begin{split}
    \Bigl\lvert\braket{\W(t,s)\Phi}{(N+1)^\delta\Bigl(\W(t,s) -\W_\nu(t,s)\Bigr)\Phi}\Bigr\rvert\leq L_\delta(t,s)\lambda^{-6\delta}\exp\Bigl\{\ass{\delta} c_1(\delta) \lambda\norm{f_0}_2 \ass{t-s}\Bigr\}\\\biggl\lvert\int_s^t\ide{\tau}\norm{\Phi}_{4\delta}\norm{\Bigl(\mathbb{R}_\nu\bigl(\widetilde{H}_{I}(\tau)+\V(\tau)\bigr)\mathbb{R}_\nu-\widetilde{H}_{I}(\tau)-\V(\tau) \Bigr)\W_\nu(\tau,s)\Phi}\biggr\rvert;
  \end{split}
\end{equation*}
by proposition~\ref{sec:useful-estimates-2}. To evaluate the last norm use the fact that for every $j$
\begin{equation*}
  (1-\mathbb{R}_\nu)\leq\frac{(N+1)^{2j}}{\sqrt{\nu}^{4j}}\leq\frac{(N+1)^{2j}}{\sqrt{\nu-1}^{4j}}
\end{equation*}
and then lemma~\ref{sec:useful-estimates-1} to obtain:
\begin{equation*}\begin{split}
      \Bigl\lvert\braket{\W(t,s)\Phi}{(N+1)^\delta\Bigl(\W(t,s) -\W_\nu(t,s)\Bigr)\Phi}\Bigr\rvert\leq L_\delta(t,s) \exp\Bigl\{\ass{\delta} c_1(\delta) \lambda\norm{f_0}_2 \ass{t-s}\Bigr\}\\\biggl\lvert\int_s^t\ide{\tau}\biggl(\norm{f_0}_2\Bigl(\frac{1}{\lambda\sqrt{\nu-1}}\Bigr)^{6\delta-1}
+\bigl(8\norm{v_{--}(\tau)}_2+2\norm{\mathcal{F}^{-1}(\chi)*A(\tau)}_\infty \bigr)\Bigl(\frac{1}{\lambda\sqrt{\nu-1}}\Bigr)^{6\delta} \biggr)\\
\exp\biggl\{\sqrt{\nu}\lambda\ass{3\delta+3/2}{c}_3(3\delta+3/2)\norm{f_0}_2\ass{\tau-s}+\ass{3\delta+3/2}\biggl(\ln 3\\+c_4(3\delta+3/2)\ass{\int_s^\tau\ide{\tau'}\norm{v_{--}(\tau')}_2}\biggr)\biggr\}\biggr\rvert
\norm{\Phi}_{4\delta}\norm{\Phi}_{6\delta+3}\; .
\end{split}\end{equation*}
The last term is easier to estimate, use again lemma~\ref{sec:useful-estimates-1} and the standard estimates for $H_{I}$ and $\V$:
\begin{equation*}\begin{split}
\ass{\braket{\Bigl(\W(t,s) -\W_\nu(t,s)\Bigr)\Phi}{(N+1)^\delta \W_\nu(t,s)\Phi}}\leq \exp\biggl\{\sqrt{\nu}\lambda\ass{\delta}{c}_3(\delta)\norm{f_0}_2\ass{t-s}\\+\ass{\delta}\biggl(\ln 3+{c}_4(\delta)\ass{\int_s^t\ide{\tau}\norm{v_{--}(\tau)}_2}\biggr)\biggr\}\biggl\lvert\int_s^t\ide{\tau}\Bigl(\lambda\norm{f_0}_2+4\norm{v_{--}(\tau)}_2+\norm{\mathcal{F}^{-1}(\chi)*A(\tau)}_\infty\Bigr)\\
\exp\biggl\{\sqrt{\nu}\lambda 2{c}_3(2)\norm{f_0}_2\ass{\tau-s}+2\biggl(\ln 3+{c}_4(2)\ass{\int_s^\tau\ide{\tau'}\norm{v_{--}(\tau')}_2}\biggr)\biggr\}\biggr\rvert\norm{\Phi}_{2\delta}\norm{\Phi}_{\mathscr{H}^4}\; .
\end{split}\end{equation*}
Fix $\nu=1+1/\lambda^2$ to complete the proof. We then obtain the following constants:
\begin{equation*}
  \begin{split}
      \mathcal{L}_1(t,s)=1+c_6(\delta) \bigl(1+2\norm{u(t)}_2^2+2\norm{\alpha(t)}_2^2\bigr)^{\delta}\bigl(1+2\norm{u(s)}_2^2+2\norm{\alpha(s)}_2^2\bigr)^{2\delta}\\\biggl(\norm{f_0}_2\ass{t-s}+\ass{\int_s^t\ide{\tau}(\norm{v_{--}(\tau)}_2+\norm{\mathcal{F}^{-1}(\chi)*A(\tau)}_\infty)}\biggr)
  \end{split}
\end{equation*}
\begin{equation*}
  \mathcal{L}_2(t,s)=c_7(\delta)\biggl(\norm{f_0}_2\ass{t-s}+\ass{\int_s^t\ide{\tau}\norm{v_{--}(\tau)}_2}+1\biggr)
\end{equation*}
\begin{equation*}
  \mathcal{C}_1=c_8(\delta)\norm{f_0}_2\; ,
\end{equation*}
with $c_6(\delta)=2^{\delta+3}c_5(2\delta)c_5(\delta)$, $c_7(\delta)=\ass{3\delta+3/2}c_4(3\delta+3/2)+20$ and $c_8(\delta)=\ass{6\delta+3}c_1(3\delta+3/2)+9$.
\end{proof}
\begin{corollary*}
  $W(t,s)$ maps $\mathscr{H}^{\delta^*}$ into $\mathscr{H}^{\delta}$. The same estimate as for $\W(t,s)$ holds:
  \begin{equation*}
  \norm{W(t,s)\Phi}^2_\delta\leq \Bigl(\lambda\norm{f_0}_2\ass{t-s}+\mathcal{L}_1(t,s)\Bigr)e^{\lambda\mathcal{C}_1\ass{t-s}+\mathcal{L}_2(t,s)}\norm{\Phi}^2_{\delta^*}\; .
  \end{equation*}
\end{corollary*}

The $\U_2$-evolution does not preserve the number of particles, however the evolution of quantum fields applied to the vacuum remains a state with only one particle. Using this fact we will be able to improve the convergence of creation and annihilation operators. We define $(\mathbb{N})_1$ to be the orthogonal projector onto $\mathscr{H}_{0,1}\oplus\mathscr{H}_{1,0}$.
\begin{proposition}\label{sec:useful-estimates-4}
Let $\mathbf{g}=\{g_i\}_{i=1}^4$ be four $L^2(\mathds{R}^3)$ functions, and consider the field $\varphi(\mathbf{g})=\psi^*(g_1)+\psi(g_2)+a^*(g_3)+a(g_4)$. Then
\begin{equation*}
  \U^*_2(t,s)\varphi(\mathbf{g})\U_2(t,s)\Omega=(\mathbb{N})_1\U^*_2(t,s)\varphi(\mathbf{g})\U_2(t,s)\Omega\; .
\end{equation*}
\end{proposition}
\begin{proof}
  Let $\Theta\in (\H_{0,1}\oplus \H_{1,0})^\perp$, and define:
  \begin{equation*}\begin{split}
      X(t)&=\sup_{g_1\in L^2}\frac{1}{\norm{g_1}_2}\ass{\braket{\Theta}{\U^*_2(t,s)\psi^*(g_1)\U_2(t,s)\Omega}}+\sup_{g_2\in L^2}\frac{1}{\norm{g_2}_2}\Bigl\lvert\braket{\Theta}{\U^*_2(t,s)\psi(g_2)\U_2(t,s)\Omega}\Bigr\rvert\\&+\sup_{g_3\in L^2}\frac{1}{\norm{g_3}_2}\ass{\braket{\Theta}{\U^*_2(t,s)a^*(g_3)\U_2(t,s)\Omega}}+\sup_{g_4\in L^2}\frac{1}{\norm{g_4}_2}\ass{\braket{\Theta}{\U^*_2(t,s)a(g_4)\U_2(t,s)\Omega}}\; .
  \end{split}\end{equation*}
If $X(t)=0$ then
\begin{gather*}
  \ass{\braket{\Theta}{\U^*_2(t,s)\varphi(\mathbf{g})\U_2(t,s)\Omega}}\leq \sup_{i\in\{1,2,3,4\}}\norm{g_i}_2X(t)=0\; ;
\end{gather*}
for all $\Theta\in (\H_{0,1}\oplus \H_{1,0})^\perp$, so $\U^*_2(t,s)\varphi(\mathbf{g})\U_2(t,s)\Omega\in \H_{0,1}\oplus \H_{1,0}$. We need to show $X(t)\leq C\int_s^t\ide{\tau}X(\tau)$: we prove it only for the first term of $X(t)$, the others being analogous. Define
\begin{equation*}
  X_1(t)=\sup_{g_1\in L^2}\frac{1}{\norm{g_1}_2}\ass{\braket{\Theta}{\U^*_2(t,s)\psi^*(g_1)\U_2(t,s)\Omega}}\; .
\end{equation*}
Differentiation yields
\begin{equation*}\begin{split}
  i\partial_t\braket{\Theta}{\U^*_2(t,s)\psi^*(g_1)\U_2(t,s)\Omega}=\braket{\Theta}{\U^*_2(t,s)[\psi^*(g_1),\V(t)]\U_2(t,s)\Omega}\; .
\end{split}\end{equation*}
Performing the commutation, integrating and taking the absolute value we obtain
\begin{equation*}\begin{split}
    \ass{\braket{\Theta}{\U^*_2(t,s)\psi^*(g_1)\U_2(t,s)\Omega}}\leq\int_s^t\ide{\tau}\Bigl\lvert\braket{\Theta}{\U^*_2(\tau,s)\Bigl(a(g_{1-})+a^*(g_{1+})+\psi^*(g_{10})\Bigr)\U_2(\tau,s)\Omega}\Bigr\rvert
\end{split}\end{equation*}
with $g_{1-}(t,\cdot)=\int\ide{x}g_1(x)\v_{--}(t,x,\cdot)$; $g_{1+}(t,\cdot)=\int\ide{x}g_1(x)\v_{-+}(t,x,\cdot)$ and $g_{10}(t,\cdot)=U_{01}^*(t)(\mathcal{F}^{-1}(\chi)*A(t))(\cdot)U_{01}(t)g_1(\cdot)$. Multiply now both members by $\norm{g_1}_2^{-1}$, and calculate the supremum in $g_1$:
\begin{align*}
  X_1(t)\leq &\int_s^t\ide{\tau}\sup_{g_{1-}\in L^2}\frac{\norm{g_{1-}}_2}{\norm{g_1}_2}\frac{1}{\norm{g_{1-}}_2}\Bigl\lvert\braket{\Theta}{\U^*_2(\tau,s)a(g_{1-}) \U_2(\tau,s)\Omega}\Bigr\rvert\\
+&\int_s^t\ide{\tau}\sup_{g_{1+}\in L^2}\frac{\norm{g_{1+}}_2}{\norm{g_1}_2}\frac{1}{\norm{g_{1+}}_2}\Bigl\lvert\braket{\Theta}{\U^*_2(\tau,s) a^*(g_{1+}) \U_2(\tau,s)\Omega}\Bigr\rvert \\
+&\int_s^t\ide{\tau}\sup_{g_{10}\in L^2}\frac{\norm{g_{10}}_2}{\norm{g_1}_2}\frac{1}{\norm{g_{10}}_2}\Bigl\lvert\braket{\Theta}{\U^*_2(\tau,s) \psi^*(g_{10}) \U_2(\tau,s)\Omega}\Bigr\rvert\; .
\end{align*}
The following estimates $\norm{g_{1-}(t)}_2\leq \norm{f_0}_2\norm{u(t)}_2\norm{g_1}_2$; $\norm{g_{1+}(t)}_2\leq \norm{f_0}_2\norm{u(t)}_2\norm{g_1}_2$ and $\norm{g_{10}(t)}_2\leq \norm{\mathcal{F}^{-1}(\chi)*A(t)}_\infty\norm{g_1}_2$ yield
\begin{gather*}
X_1(t)\leq C\int_s^t\ide{\tau}X(\tau)\; ;\; C=\sup_{\tau\in [s,t]}\Bigl( 2\norm{f_0}_2\norm{u(\tau)}_2 + \norm{\mathcal{F}^{-1}(\chi)*A(\tau)}_\infty\Bigr)\; .
\end{gather*}
\end{proof}

\subsection{Proof of Theorem~\ref{sec:class-limit-annih-1}.}
\label{sec:proof-theorem}

Let $f\in L^2(\mathds{R}^3)$, define
\begin{align*}
\langle\lambda\psi^\#(\bar{f}^\#)(t)\rangle_{C\Phi}&=\braket{C(u_\lambda,\alpha_\lambda)\Phi}{U^*(t)\lambda\psi^\#(\bar{f}^\#)U(t)C(u_\lambda,\alpha_\lambda)\Phi}\\
\langle\lambda a^\#(\bar{f}^\#)(t)\rangle_{C\Phi}&=\braket{C(u_\lambda,\alpha_\lambda)\Phi}{U^*(t)\lambda a^\#(\bar{f}^\#)U(t)C(u_\lambda,\alpha_\lambda)\Phi}\; .
\end{align*}
This definition yields:
\begin{equation*}
  \langle\lambda\psi^\#(\bar{f}^\#)(t)\rangle_{C\Phi}=\int\ide{x}\bar{f}^\#(x)\langle\lambda \psi^\#(t,x)\rangle_{C\Phi}\; ;\; \langle\lambda a^\#(\bar{f}^\#)(t)\rangle_{C\Phi}=\int\ide{k}\bar{f}^\#(k)\langle\lambda a^\#(t,k)\rangle_{C\Phi} .
\end{equation*}
\begin{lemma}
  \label{sec:gener-case:-conv-1}
Let $\Phi\in\mathscr{H}^{\delta}$, $\delta\geq 9$, $f\in L^2(\mathds{R}^3)$, $(u(\cdot),\alpha(\cdot))$ the $\mathscr{C}^0(\mathds{R},L^2(\mathds{R}^3)\otimes L^2(\mathds{R}^3))$ solution of~\eqref{eq:2} with initial conditions $(u,\alpha)\in L^2(\mathds{R}^3)\otimes L^2(\mathds{R}^3)$. Then
\begin{align*}
  \ass{\langle\lambda\psi^\#(\bar{f}^\#)(t)\rangle_{C\Phi}-\braket{f^\#}{u^\#(t)}_2}\leq& \lambda \norm{f}_2 \sqrt{\lambda\norm{f_0}_2\ass{t}+\mathcal{L}_1(t,0)}e^{\frac{\lambda\mathcal{C}_1\ass{t}+\mathcal{L}_2(t,0)}{2}}\norm{\Phi}_{\mathscr{H}^9}\\
  \ass{\langle\lambda a^\#(\bar{f}^\#)(t)\rangle_{C\Phi}-\braket{f^\#}{\alpha^\#(t)}_2}\leq& \lambda \norm{f}_2 \sqrt{\lambda\norm{f_0}_2\ass{t}+\mathcal{L}_1(t,0)}e^{\frac{\lambda\mathcal{C}_1\ass{t}+\mathcal{L}_2(t,0)}{2}}\norm{\Phi}_{\mathscr{H}^9}
\end{align*}
where the constants are defined in the proof of proposition~\ref{sec:useful-estimates-3}.
\end{lemma}
\begin{proof}
  We prove the result for $\langle\lambda\psi^\#(\bar{f}^\#)(t)\rangle_{C\Phi}$, the other case being perfectly analogous. Proposition~\ref{sec:quantum-theory-8} yields $\langle\lambda\psi^\#(\bar{f}^\#)(t)\rangle_{C\Phi}=\braket{W(t,0)\Phi}{\lambda\psi^\#(\bar{f}^\#)W(t,0)\Phi}+\braket{f^\#}{u^\#(t)}_2$. Then $\ass{\langle\lambda\psi^\#(\bar{f}^\#)(t)\rangle_{C\Phi}-\braket{f^\#}{u^\#(t)}_2}\leq \lambda\norm{\psi^\#(\bar{f}^\#)W(t,0)\Phi}\leq \lambda\norm{f}_2\norm{W(t,0)\Phi}_{\mathscr{H}^{1}}$. Apply corollary of Proposition~\ref{sec:useful-estimates-3} to obtain the result.
\end{proof}
Since the bound of the lemma above holds for all $f\in L^2(\mathds{R}^3)$, the Lemma of Riesz implies $\langle\lambda \psi^\#(t,x)\rangle_{C\Phi}$, $\langle\lambda a^\#(t,k)\rangle_{C\Phi}\in L^2(\mathds{R}^3)$. Furthermore they satisfy the bounds stated in the theorem. If $\Phi=\Omega$, we can apply proposition~\ref{sec:useful-estimates-4}:
  \begin{lemma}\label{sec:proof-theor-refs}
    Let $f\in L^2(\mathds{R}^3)$, then we have the following bounds:
\begin{align*}
  \ass{\langle\lambda \psi^\#(\bar{f}^\#)(t)\rangle_{C\Omega}-\braket{f^\#}{u^\#(t)}_2}\leq& \lambda^2 \norm{f}_2 K_1 \ass{t}e^{K_2\ass{t}}\\
  \ass{\langle\lambda a^\#(\bar{f}^\#)(t)\rangle_{C\Omega}-\braket{f^\#}{\alpha^\#(t)}_2}\leq& \lambda^2 \norm{f}_2 K_1\ass{t}e^{K_2\ass{t}}
\end{align*}
with $K_1$ and $K_2$ positive constants. 
  \end{lemma}
  \begin{proof}
    As usual we prove the result for $\langle\lambda \psi^\#(\bar{f}^\#)(t)\rangle_{C\Omega}$, the other case being perfectly analogous. Using proposition~\ref{sec:quantum-theory-8} we can write
\begin{equation*}
  \begin{split}
    \ass{\langle\lambda \psi^\#(\bar{f}^\#)(t)\rangle_{C\Omega}-\braket{f^\#}{u^\#(t)}_2}=\lambda\ass{\braket{\W(t,0)\Omega}{U_0^*(t)\psi^\#(\bar{f}^\#)U_0(t)\W(t,0)\Omega}}\; .
  \end{split}
\end{equation*}
Then defining $\tilde{f}\equiv U_{01}^*(t)f$:
\begin{equation*}
    \ass{\langle\lambda \psi^\#(\bar{f}^\#)(t)\rangle_{C\Omega}-\braket{f^\#}{u^\#(t)}_2}=\lambda\ass{\braket{\W(t,0)\Omega}{\psi^\#(\bar{\tilde{f}}^\#)\W(t,0)\Omega}}\; .
\end{equation*}
By equality
\begin{equation*}\begin{split}
  \braket{\W(t,0)\Omega}{\psi^\#(\bar{\tilde{f}}^\#)\W(t,0)\Omega}=\braket{\U_2(t,0)\Omega}{\psi^\#(\bar{\tilde{f}}^\#)\U_2(t,0)\Omega}+\braket{(\W(t,0)-\U_2(t,0))\Omega}{\\\psi^\#(\bar{\tilde{f}}^\#)\U_2(t,0)\Omega}+\braket{\W(t,0)\Omega}{\psi^\#(\bar{\tilde{f}}^\#)(\W(t,0)-\U_2(t,0))\Omega}
\end{split}\end{equation*}
write
\begin{equation*}\begin{split}
    \ass{\langle\lambda \psi^\#(\bar{f}^\#)(t)\rangle_{C\Omega}-\braket{f^\#}{u^\#(t)}_2}\leq\lambda\Bigl(\ass{\braket{\U_2(t,0)\Omega}{\psi^\#(\bar{\tilde{f}}^\#)\U_2(t,0)\Omega}}\\
+\ass{\braket{(\W(t,0)-\U_2(t,0))\Omega}{\psi^\#(\bar{\tilde{f}}^\#)\U_2(t,0)\Omega}}
+\ass{\braket{\W(t,0)\Omega}{\psi^\#(\bar{\tilde{f}}^\#)(\W(t,0)-\U_2(t,0))\Omega}}\Bigr)\\\equiv \lambda\Bigl(X_1+X_2+X_3\Bigr)\; .
\end{split}\end{equation*}
By proposition~\ref{sec:useful-estimates-4} we have that $X_1=0$. Then we bound $X_2$ as follows:
\begin{equation*}
  \begin{split}
    X_2=\ass{\braket{\int_0^t\ide{\tau}\dtot{}{\tau}\W^*(\tau,0)\U_2(\tau,0)\Omega}{\W^*(t,0)\psi^\#(\bar{\tilde{f}}^\#)\U_2(t,0)\Omega}}\leq \lambda\norm{f}_2\norm{f_0}_2\\\ass{\int_0^t\ide{\tau}\norm{\U_2(\tau,0)\Omega}_{\mathscr{H}^4}\norm{\U_2(t,0)\Omega}_{\mathscr{H}^1}}
  \end{split}
\end{equation*}
Using proposition~\ref{sec:evol-quant-fluct-6} we obtain
\begin{equation*}
  \begin{split}
    X_2\leq \lambda\norm{f}_2\norm{f_0}_2\exp\biggl\{\frac{1}{2}\biggl(\ln 3+ 4\sqrt{2}\ass{\int_0^t\ide{\tau}\norm{v_{--}(\tau)}_2}\biggr)\biggr\}\biggl\lvert\int_0^t\ide{\tau}\exp\biggl\{2\biggl(\ln 3\\+ 10\sqrt{2}\ass{\int_0^\tau\ide{\tau'}\norm{v_{--}(\tau')}_2}\biggr)\biggr\}\biggr\rvert\; .
  \end{split}
\end{equation*}
To bound $X_3$ we use a similar method:
\begin{equation*}
  \begin{split}
    X_3\leq \lambda\norm{f}_2\norm{f_0}_2 \Bigl(\lambda\norm{f_0}_2\ass{t}+\mathcal{L}_1(t,0)\Bigr)^{1/2}e^{(\lambda\mathcal{C}_1\ass{t}+\mathcal{L}_2(t,0))/2}\biggl\lvert\int_0^t\ide{\tau}\exp\biggl\{2\biggl(\ln 3\\
+ 10\sqrt{2}\ass{\int_0^\tau\ide{\tau'}\norm{v_{--}(\tau')}_2}\biggr)\biggr\}\biggr\rvert\; .
  \end{split}
\end{equation*}
\end{proof}

\subsection{Proof of Theorem~\ref{sec:class-limit-annih-2}.}
\label{sec:proof-th}

The most difficult case is with $\Theta$ vectors. We write the proof in that case, the other being analogous. We recall the definition of operator $B$, and define transition amplitudes $\langle B \rangle_\Theta(t)$.

    Let $q,r,h,l\in \mathds{N}$, $\delta=q+r+h+l$, $g\in L^2(\mathds{R}^{3(q+r)})\otimes L^2(\mathds{R}^{3(h+l)})\equiv L^2(\mathds{R}^{3\delta})$. Then
\begin{equation*}\begin{split}
  B=\int\de X_q\de Y_r\de K_h\ide{M_l} \bar{g}(x_1,\dotsc,x_q;y_1,\dotsc,y_r;k_1,\dotsc,k_h;m_1,\dotsc,m_l)\\
\psi^*(X_q)\psi(Y_r)a^*(K_h)a(M_l)
\end{split}\end{equation*}
\begin{equation*}
  \langle B \rangle_\Theta(t)\equiv \lambda^\delta\braket{\Theta(t)}{B\Theta(t)}=\lambda^\delta\braket{\Theta}{U^*(t)B U(t)\Theta}\; .  
\end{equation*}
Now let $0\leq d\leq\delta$, $(u_\theta(t),\alpha_\theta(t))$ the $\mathscr{C}^0(\mathds{R},L^2(\mathds{R}^3)\otimes L^2(\mathds{R}^3))$ solution of~\eqref{eq:2} corresponding to initial data $(u_0,\alpha_0(\theta)\equiv\exp\{-i\theta\}\alpha_0)$. We establish the following correspondence:
  \begin{equation*}
    \psi(x)\longleftrightarrow \frac{1}{\lambda}u_\theta(t,x)\; ; \; \psi^*(x)\longleftrightarrow \frac{1}{\lambda}\bar{u}_\theta(t,x)\; ;\; a(k)\longleftrightarrow \frac{1}{\lambda}\alpha_\theta(t,k)\; ;\; a^*(k)\longleftrightarrow \frac{1}{\lambda}\bar{\alpha}_\theta(t,k)\; .
  \end{equation*}
We will call $B^{(d)}$ the operator obtained substituting in any possible way $d$ creation or annihilation operators of $B$ with functions, following the correspondence above. $B^{(d)}$ is the sum of $\binom{\delta}{d}$ operators of type $B$, but with $\delta-d$ creation or annihilation operators.
\begin{lemma}\label{sec:general-remarks-1}
  For any $0\leq d\leq\delta$ exists a function $C_d(t)$, depending on $\norm{u_\theta(t)}_2$ and $\norm{\alpha_\theta(t)}_2$, such that for all $\Phi\in \mathscr{H}^{\delta-d}$ we have the following inequality:
  \begin{equation*}\begin{split}
    \norm{B^{(d)}\Phi}\leq \lambda^{-d}(q+h)^{(\delta-d)/2}C_d(t)\norm{g;L^2(\mathds{R}^{3\delta})}\norm{\Phi}_{\delta-d} \; .
 \end{split} \end{equation*}
\end{lemma}
\begin{proof}
  The proof is a simple application of lemma~\ref{sec:quantum-theory-5}. $C_d(t)$ would be the sum of products of the $L^2$-norms of $u_\theta(t)$ and $\alpha_\theta(t)$, for example $C_1(t)=(q+r)\norm{u_\theta(t)}_2+(h+l)\norm{\alpha_\theta(t)}_2$.
\end{proof}
For all $\Phi\in\mathscr{H}^\delta$ the following identity holds:
  \begin{equation}\label{eq:7}
    B'\Phi\equiv C^*(u(t)/\lambda,\alpha(t)/\lambda)B C(u(t)/\lambda,\alpha(t)/\lambda)\Phi=\sum_{d=0}^\delta B^{(d)}\Phi\; .
  \end{equation}
For any $h\in \mathscr{C}^0(\mathds{R},L^2(\mathds{R}^3))$, $q\in\mathds{N}$ we define
\begin{equation*}
  h^{\otimes_{q}}(t,X_q)\equiv h(t,x_1)\otimes \dotsm\otimes h(t,x_q)\in  \mathscr{C}^0(\mathds{R},L^2(\mathds{R}^{3q}))\; .
\end{equation*}
Then we can formulate the following lemma.
\begin{lemma}\label{sec:general-remarks-3}
  Let $\varphi(\mathbf{g})$ defined as in Proposition~\ref{sec:useful-estimates-4}. Then for all $\Phi\in\mathscr{H}^\delta$ the following equality holds:
  \begin{equation*}
    B^{(\delta-1)}\Phi=\lambda^{-\delta+1}\varphi(\mathbf{g})\Phi\; ,
  \end{equation*}
with
\begin{align*}
  g_1(x)=& \sum_{\alpha=1}^q \int\de x_1\dotsm\de \hat{x}_\alpha\dotsm \de x_q\de{Y_r}\de{K_h}\ide{M_l}g(\dotsc,x_{\alpha-1},x,x_{\alpha+1},\dotsc)\\
&\bar{u}_\theta^{\otimes_{q-1}}(t,X_q\setminus x_\alpha)u_\theta^{\otimes_r}(t,Y_r)\bar{\alpha}_\theta^{\otimes_h}(t,K_h)\alpha_\theta^{\otimes_l}(t,M_l)\\
  g_2(x)=& \sum_{\alpha=1}^r \int\de{X_q}\de y_1\dotsm \de \hat{y}_\alpha\dots\de y_r\de{K_h}\ide{M_l}g(\dotsc,y_{\alpha-1},x,y_{\alpha+1},\dotsc)\\
&\bar{u}_\theta^{\otimes_{q}}(t,X_q)u_\theta^{\otimes_{r-1}}(t,Y_r\setminus y_\alpha)\bar{\alpha}_\theta^{\otimes_h}(t,K_h)\alpha_\theta^{\otimes_l}(t,M_l)
\end{align*}
\begin{align*}
  g_3(k)=& \sum_{\alpha=1}^h \int\de{X_q}\de{X_r}\de k_1\dotsm\de\hat{k}_\alpha\dotsm \de k_h\ide{M_l}g(\dotsc,k_{\alpha-1},k,k_{\alpha+1},\dotsc)\\
&\bar{u}_\theta^{\otimes_{q}}(t,X_q)u_\theta^{\otimes_{r}}(t,Y_r)\bar{\alpha}_\theta^{\otimes_{h-1}}(t,K_h\setminus k_\alpha)\alpha_\theta^{\otimes_l}(t,M_l)\\
  g_4(k)=& \sum_{\alpha=1}^l \int\de{X_q}\de{Y_r}\de{K_h}\de m_1\dotsm\de \hat{m}_\alpha\dotsm\ide{m_l}g(\dotsc,m_{\alpha-1},k,m_{\alpha+1},\dotsc)\\
&\bar{u}_\theta^{\otimes_{q}}(t,X_q)u_\theta^{\otimes_{r}}(t,Y_r)\bar{\alpha}_\theta^{\otimes_h}(t,K_h)\alpha_\theta^{\otimes_{l-1}}(t,M_l\setminus m_\alpha)\; .
\end{align*}
\end{lemma}
To improve readability we make the following definitions:
  \begin{equation*}
    L_W(\delta,t,s)=\Bigl(\lambda\norm{f_0}_2\ass{t-s}+\mathcal{L}_1(\delta,t,s)\Bigr)e^{\lambda\mathcal{C}_1(\delta)\ass{t-s}+\mathcal{L}_2(\delta,t,s)}
  \end{equation*}
where the functions and constants on the right hand side are defined in Proposition~\ref{sec:useful-estimates-3}, with $\delta$-dependence made explicit.
\begin{equation*}
  L_U(\delta,t,s)=\exp\biggl\{\frac{\ass{\delta}}{2}\biggl(\ln 3+ \sqrt{2}c_2(\delta)\ass{\int_s^t\ide{\tau}\norm{v_{--}(\tau)}_2}\biggr)\biggr\}\; ,
\end{equation*}
with $c_2(\delta)=\max(4,3^{\ass{\delta}/2}+1)$.
\begin{proposition}
\label{sec:proof-psi-vectors-1}
Two constants $K_j(\Theta)$ with $j=1,2$ exist such that for all $g\in L^2(\mathds{R}^{3\delta})$
\begin{equation*}\begin{split}
\ass{\langle B \rangle_\Theta(t)-\delta_{qr}\int_0^{2\pi}\frac{\de\theta}{2\pi}\braket{g}{\bar{u}^{\otimes_q}_\theta u^{\otimes_r}_\theta\bar{\alpha}^{\otimes_i}_\theta\alpha^{\otimes_j}_\theta(t)}_{L^2(\mathds{R}^{3\delta})}}\leq \delta_{qr}\lambda^2 \norm{g;L^2(\mathds{R}^{3\delta})}K_1(\Theta)\ass{t}e^{K_2(\Theta)\ass{t}}\; .
\end{split}\end{equation*}
\end{proposition}
\begin{proof}
We need the following lemma, proved applying results proved in~\citet{2011arXiv1103.0948C} and~\citet{MR2530155} to $\mathscr{F}_s(p)$ and $\mathscr{F}_s(n)$.
\begin{lemma}\label{sec:trans-ampl-8}
Let $u_0,\alpha_0\in L^2(\mathds{R}^3)$, such that $\norm{u_0}_2=\norm{\alpha_0}_2=1$ and $d_{x}\equiv\frac{\sqrt{x!}}{e^{-x/2}x^{x/2}}$. Then we obtain the following identities:
\begin{gather*}
\Theta=d^2_{\lambda^{-2}}(\mathbb{N}_1)_{\lambda^{-2}}\int_0^{2\pi}\frac{\de\theta}{2\pi}e^{i\lambda^{-2}\theta}C(u_0/\lambda,\alpha_0(\theta)/\lambda)\Omega=d^2_{\lambda^{-2}}(\mathbb{N}_1)_{\lambda^{-2}}(\mathbb{N}_2)_{\lambda^{-2}}C(u_0/\lambda,\alpha_0/\lambda)\Omega\\
(\mathbb{N})_1C^*(u_0/\lambda,\alpha_0(\theta)/\lambda)\Theta=0\; ,
\end{gather*}
where $(\mathbb{N}_1)_p$ is the orthogonal projector on $\mathscr{H}_p$, $(\mathbb{N}_1)_{p}(\mathbb{N}_2)_{n}$ the projector on $\mathscr{H}_{p,n}$, $(\mathbb{N})_1$ the projector on $\mathscr{H}_{0,1}\oplus\mathscr{H}_{1,0}$.
\end{lemma}
Write $\langle B \rangle_\Theta(t)=\lambda^\delta\braket{\Theta}{U^*(t)B U(t)\frac{\psi^*(u_0)^{\lambda^{-2}}}{\sqrt{\lambda^{-2}!}}\frac{a^*(\alpha_0)^{\lambda^{-2}}}{\sqrt{\lambda^{-2}!}}\Omega}$. Observe that when $q\neq r$ we have $\langle B \rangle_\Theta(t)=0$ since $N_1$ commutes with $H$ and $B$ doesn't preserve the number of non-relativistic particles. So we will set $q=r$ for the rest of the proof. Using equation~\eqref{eq:7} and the first equality of lemma~\ref{sec:trans-ampl-8} we obtain
\begin{equation*}\begin{split}
    \langle B \rangle_\Theta(t)=\lambda^\delta d^2_{\lambda^{-2}}\int_0^{2\pi}\frac{\de\theta}{2\pi}e^{i\theta/\lambda^2}\braket{C^*(u_0/\lambda,\alpha_0(\theta)/\lambda)\Theta}{W_\theta^*(t,0)B' W_\theta(t,0)\Omega}\; ,
\end{split}\end{equation*}
where $W_\theta$ is the operator $W$ defined above with solution $(u_\theta(t),\alpha_\theta(t))$ instead of $(u(t),\alpha(t))$. Since throughout the proof we will use $W_\theta$ instead of $W$, we will omit the index $_\theta$. First of all consider $B^{(\delta)}$:
\begin{equation*}
  \begin{split}
    \langle B^{(\delta)} \rangle_\Theta(t)=\int_0^{2\pi}\frac{\de\theta}{2\pi}\braket{g}{\bar{u}^{\otimes_q}_\theta u^{\otimes_r}_\theta\bar{\alpha}^{\otimes_i}_\theta\alpha^{\otimes_j}_\theta(t)}_{L^2(\mathds{R}^{3\delta})}\braket{\Theta}{d^2_{\lambda^{-2}}e^{i\theta/\lambda^2} C(u_0/\lambda,\alpha_0(\theta)/\lambda)\Omega}
  \end{split}
\end{equation*}
then since $\Theta\in\mathscr{H}_{\lambda^{-2},\lambda^{-2}}$ and $\alpha_0(\theta)=\exp\{-i\theta\}\alpha_0$, the second equality of lemma~\ref{sec:trans-ampl-8} yields
\begin{equation*}
  \begin{split}
    \langle B^{(\delta)} \rangle_\Theta(t)=\int_0^{2\pi}\frac{\de\theta}{2\pi}\braket{g}{\bar{u}^{\otimes_q}_\theta u^{\otimes_r}_\theta\bar{\alpha}^{\otimes_i}_\theta\alpha^{\otimes_j}_\theta(t)}_{L^2(\mathds{R}^{3\delta})}\; .
  \end{split}
\end{equation*}
Then
\begin{equation*}\begin{split}
 \langle B \rangle_\Theta(t)-\int_0^{2\pi}\frac{\de\theta}{2\pi}\braket{g}{\bar{u}^{\otimes_q}_\theta u^{\otimes_r}_\theta\bar{\alpha}^{\otimes_i}_\theta\alpha^{\otimes_j}_\theta(t)}_{L^2(\mathds{R}^{3\delta})}=\sum_{d=0}^{\delta-1}\lambda^{\delta-d} d^2_{\lambda^{-2}}\int_0^{2\pi}\frac{\de\theta}{2\pi}e^{i\theta/\lambda^2}\\
\braket{\frac{1}{\sqrt{(N_1+1)(N_2+1)}}C^*(u_0/\lambda,\alpha_0(\theta)/\lambda)\Theta}{\sqrt{(N_1+1)(N_2+1)}W^*(t,0)B^{(d)}W(t,0)\Omega}.\end{split}
\end{equation*}
The following lemma is proved in~\citet{2011JMP} using sharp estimates of Laguerre polynomials obtained by~\citet{MR2168916}.
\begin{lemma}\label{sec:trans-ampl-7}
Let $u_0,\alpha_0\in L^2(\mathds{R}^3)$ such that $\norm{u_0}_2=\norm{\alpha_0}_2=1$. Then there is a constant $L_\Theta$ independent of $\lambda$ and $\theta$ such that:
  \begin{gather*}
    \norm{(N_1+1)^{-1/2}(N_2+1)^{-1/2}C^*(u_0/\lambda,\alpha_0(\theta)/\lambda)\Theta}\leq L_\Theta d_{\lambda^{-2}}^{-2}\; .
  \end{gather*}
\end{lemma}
Using it we obtain
\begin{equation*}
  \begin{split}
    \ass{ \langle B \rangle_\Theta(t)-\int_0^{2\pi}\frac{\de\theta}{2\pi}\braket{g}{\bar{u}^{\otimes_q}_\theta u^{\otimes_r}_\theta\bar{\alpha}^{\otimes_i}_\theta\alpha^{\otimes_j}_\theta(t)}_{L^2(\mathds{R}^{3\delta})}}\leq\sum_{d=0}^{\delta-2}\lambda^{\delta-d}L_\Theta\int_0^{2\pi}\frac{\de\theta}{2\pi}\\\norm{W^*(t,0)B^{(d)}W(t,0)\Omega}_{\mathscr{H}^{2}}
+\lambda^{\delta}d^2_{\lambda^{-2}}\int_0^{2\pi}\frac{\de\theta}{2\pi}\Bigl\lvert\braket{C^*(u_0/\lambda,\alpha_0(\theta)/\lambda)\Theta}{\\W^*(t,0)B^{(\delta -1)}W(t,0)\Omega}\Bigr\rvert\; .
  \end{split}
\end{equation*}

Since we are interested in the region where $\lambda<1$, $\lambda^a\leq\lambda^2$ for any $a\geq 2$. Furthermore $\norm{\Omega}_\delta=1$ for any $\delta$. We can apply two times the corollary of proposition~\ref{sec:useful-estimates-3}, lemma~\ref{sec:quantum-theory-5} and lemma~\ref{sec:general-remarks-1} to obtain: 
\begin{equation*}
  \begin{split}
    \norm{W^*(t,0)B^{(d)}W(t,0)\Omega}_{\mathscr{H}^2}\leq(q+h+1)^{\frac{15+\delta-d}{2}}C_d(t)L_W(2,0,t)L_W(\delta-d+2,t,0)\\
\norm{g;L^2(\mathds{R}^{3\delta})}
  \end{split}
\end{equation*}
So we can write:
\begin{equation*}
  \begin{split}
    \ass{ \langle B \rangle_\Psi(t)-\braket{g}{\bar{u}^{\otimes_q}u^{\otimes_q}\bar{\alpha}^{\otimes_i}\alpha^{\otimes_j}(t)}_{L^2(\mathds{R}^{3\delta})}}\leq \lambda^2 K_1\ass{t}e^{K_2\ass{t}}\norm{g;L^2(\mathds{R}^{3\delta})}\\
+\lambda^{\delta}d^2_{\lambda^{-2}}\int_0^{2\pi}\frac{\de\theta}{2\pi}\ass{\braket{C^*(u_0/\lambda,\alpha_0/\lambda)\Theta}{W^*(t,0)B^{(\delta -1)}W(t,0)\Omega}}\; ,
  \end{split}
\end{equation*}
with
\begin{equation*}
  \begin{split}
    K_1\ass{t}e^{K_2\ass{t}}\geq \sum_{d=0}^{\delta-2}(q+h+1)^{\frac{15+\delta-d}{2}}L_\Theta\int_0^{2\pi}\frac{\de\theta}{2\pi}C_d(t) L_W(2,0,t)L_W(\delta-d+2,t,0)\; .
  \end{split}
\end{equation*}

We have to use a different approach to estimate the last term of the inequality above, namely
\begin{equation*}\begin{split}
X\equiv \lambda^{\delta}d^2_{\lambda^{-2}}\int_0^{2\pi}\frac{\de\theta}{2\pi}\ass{\braket{C^*(u_0/\lambda,\alpha_0/\lambda)\Theta}{W^*(t,0)B^{(\delta -1)}W(t,0)\Omega}}\; .
\end{split}\end{equation*}
By Lemma~\ref{sec:general-remarks-3} and passing to the interaction representation:
\begin{equation*}
  X=\lambda^{\delta}d^2_{\lambda^{-2}}\int_0^{2\pi}\frac{\de\theta}{2\pi}\ass{\braket{C^*(u_0/\lambda,\alpha_0/\lambda)\Theta}{\W^*(t,0)\varphi(\tilde{\mathbf{g}})\W(t,0)\Omega}}\; ,
\end{equation*}
with $\tilde{g}_1(x)= U_{01}^*(t)g_1(x)$; $\tilde{g}_2(x)= U_{01}(t)g_2(x)$; $\tilde{g}_3(k)=U_{02}^*(t)g_3(k)$ and $\tilde{g}_4(k)= U_{02}(t)g_4(x)$. By the following identity:
\begin{equation*}
  \begin{split}
    \braket{\Phi}{\W^*(t,0) \varphi(\tilde{\mathbf{g}})\W(t,0)\Omega}=\braket{\Phi}{\U_2^*(t,0)\varphi(\tilde{\mathbf{g}})\U_2(t,0)\Omega}+\braket{\Phi}{(\W^*(t,0)\\-\U^*_2(t,0)) \varphi(\tilde{\mathbf{g}})\W(t,0)\Omega}+\braket{\Phi}{\U_2^*(t,0) \varphi(\tilde{\mathbf{g}})(\W(t,0)-\U_2(t,0))\Omega}\; ;
  \end{split}
\end{equation*}
with $\Phi=C^*(u_0/\lambda,\alpha_0/\lambda)\Theta$ we obtain using lemma~\ref{sec:trans-ampl-8}
\begin{equation*}\begin{split}
X\leq\lambda d^2_{\lambda^{-2}}\int_0^{2\pi}\frac{\de\theta}{2\pi}\Bigl(\ass{\braket{\Phi}{(\W^*(t,0)-\U^*_2(t,0)) \varphi(\tilde{\mathbf{g}})\W(t,0)\Omega}}+\Bigl\lvert\braket{\Phi}{\U_2^*(t,0) \varphi(\tilde{\mathbf{g}})(\W(t,0)\\-\U_2(t,0))\Omega}\Bigr\rvert\Bigr)\equiv \lambda d^2_{\lambda^{-2}}\int_0^{2\pi}\frac{\de\theta}{2\pi} (X_1+X_2)\; .
\end{split}\end{equation*}
We define $\norm{\mathbf{g}}_2=\norm{g_1}_2+\norm{g_2}_2+\norm{g_3}_2+\norm{g_4}_2$. Bound $X_1$, the integral making sense as strong Riemann integral on $\mathscr{H}$:
\begin{align*}
    X_1&\leq L_\Theta d^{-2}_{\lambda^{-2}}\biggl\lvert\int_0^t\ide{\tau}\Bigl\lVert \W^*(\tau,0)U_0^*(\tau)H_I U_0(\tau)\U_2(\tau,0)\U^*_2(t,0)\varphi(\tilde{\mathbf{g}})\W(t,0)\Omega\Bigr\rVert_{\mathscr{H}^2}\biggr\rvert\; .
\end{align*}
We remark that $\norm{\tilde{\mathbf{g}}}_2=\norm{\mathbf{g}}_2\leq C_{\delta-1}(t)\norm{g;L^2(\mathds{R}^{3\delta})}$, with $C_{\delta-1}(t)$ defined in Lemma~\ref{sec:general-remarks-1}. Then $\int_0^{2\pi}\frac{\de\theta}{2\pi}  X_1\leq \lambda d^{-2}_{\lambda^{-2}} K_1'\ass{t}e^{K'_2\ass{t}}\norm{g;L^2(\mathds{R}^{3\delta})}$ with
\begin{equation*}
  \begin{split}
    K_1'\ass{t}e^{K'_2\ass{t}} \geq L_\Theta 2^{17}\norm{f_0}_2 \int_0^{2\pi}\frac{\de\theta}{2\pi} C_{\delta-1}(t) L_U(19,0,t)L_W(20,t,0)\biggl\lvert\int_0^t\ide{\tau}L_W(2,0,\tau)\\L_U(19,\tau,0)\biggr\rvert
  \end{split}
\end{equation*}
Analogously bound $X_2$: $\int_0^{2\pi}\frac{\de\theta}{2\pi} X_2\leq \lambda d^{-2}_{\lambda^{-2}} K_1''\ass{t}e^{K''_2\ass{t}}\norm{g;L^2(\mathds{R}^{3\delta})}$ with
\begin{equation*}
  \begin{split}
    K_1''\ass{t}e^{K''_2\ass{t}}\geq L_\Theta 2^{65}\norm{f_0}_2 \int_0^{2\pi}\frac{\de\theta}{2\pi} C_{\delta-1}(t) L_U(2,0,t)L_W(3,t,0)\biggl\lvert\int_0^t\ide{\tau}L_W(21,0,\tau)\\L_U(133,\tau,0)\biggr\rvert
  \end{split}
\end{equation*}
\end{proof}
\begin{acknowledgments}
  The author wishes to express all his gratitude and give many thanks to Professor Giorgio Velo for illuminating advice and helpful discussions.
\end{acknowledgments}

\end{document}